\documentclass[abbrvnat]{sigplanconf}

\usepackage{graphicx}
\usepackage{xspace}
\usepackage{color}
\usepackage{amsthm}
\usepackage{multirow}
\usepackage{url}
\usepackage{algorithm}
\usepackage{algpseudocode}
\newcommand\techrep[1]{}
\newcommand\nottechrep[1]{#1}

\newcommand{\impl}{\textsc{mPy}}
\newcommand{\implp}{\widetilde{\impl}}

\renewcommand\t[1]{\texttt{#1}}
\newcommand\bnfas{\;:=\;}
\newcommand\bnfalt{\;|\;}
\newcommand\var{v}

\newcommand\fvar{f}
\newcommand\sem[1]{[\![#1]\!]}
\newcommand\itos[1]{\Phi(#1)}

\newcommand\prog{p}
\newcommand\aexpr{a}
\newcommand\bexpr{b}
\newcommand\sexpr{s}
\newcommand\intN{n}
\newcommand\op{op}
\newcommand\aop{op_\aexpr}
\newcommand\bop{op_\bexpr}
\newcommand\cop{op_c}
\newcommand\xop{op_x}

\newcommand\fs{\tilde{f}}

\newcommand\sprog{\tilde{p}}

\newcommand\saexpr{\tilde{\aexpr}}
\newcommand\sbexpr{\tilde{\bexpr}}
\newcommand\ssexpr{\tilde{\sexpr}}
\newcommand\sop{\widetilde{\op}}
\newcommand\saop{\widetilde{op}_a}
\newcommand\sxop{\widetilde{op}_x}
\newcommand\sbop{\widetilde{op}_b}
\newcommand\scop{\widetilde{op}_c}
\newcommand\svar{\tilde{\var}}

\newcommand\where{\;|\;}
\newcommand\choosevar{?}

\newcommand\csharp{\textsc{C\#}\xspace}

\newcommand {\Sk}{\textsc{Sketch}}
\newcommand{\fm}{\mathcal{E}}
\newcommand{\emod}{\mathcal{E}} 
\newcommand{\crule}{\mathcal{C}} 
\newcommand{\initF}{\textsc{InitR}\xspace}
\newcommand{\condF}{\textsc{CompR}\xspace}
\newcommand{\indF}{\textsc{IndR}\xspace}
\newcommand{\retF}{\textsc{RetR}\xspace}
\newcommand{\rangeF}{\textsc{RanR}\xspace}

\newcommand{\dash}{'}

\newcommand{\transExpr}{\mathcal{T}\xspace}
\newcommand{\matchE}{\t{Match}\xspace}

\newcommand{\sif}{\t{if}~}

\newcommand{\selse}{\t{else}~}

\newcommand{\wexpr}{w\xspace}

\usepackage{listings}

\usepackage[T1]{fontenc}
\usepackage[scaled=0.80]{beramono}

\usepackage{color}
\definecolor{bluekeywords}{rgb}{0.13,0.13,1}
\definecolor{greencomments}{rgb}{0,0.5,0}
\definecolor{redstrings}{rgb}{0.9,0,0}

\usepackage{listings}
\newcommand{\C}[1]{\lstinline!#1!}

\newcommand{\seclabel}[1]{\label{sec:#1}}

\newcommand{\longsecref}[1]{Section~\ref{sec:#1}}
\newcommand{\tablabel}[1]{\label{tab:#1}}
\newcommand{\tabref}[1]{Table~\ref{tab:#1}}
\newcommand{\figlabel}[1]{\label{fig:#1}}
\newcommand{\longfigref}[1]{Figure~\ref{fig:#1}}

\newcommand{\secref}[1]{\longsecref{#1}}
\newcommand{\figref}[1]{\longfigref{#1}}

\newcommand{\pexfun}{\textsc{Pex4Fun}\xspace}

\newcommand{\eml}{\textsc{Eml}\xspace}
\newcommand{\defchoice}[1]{\framebox{$#1$}}
\newcommand{\tr}{\mathcal{T}}

\newcommand\choicek{\t{choice}_\t{k}}
\newcommand\inpstate{\sigma}
\newcommand\cnstr{\phi}
\newcommand\cnstrset{\Phi}
\newcommand\cnstrprev{\phi_p}

\algdef{SxnE}[WHILE]{While}{EndWhile}[1]{\algorithmicwhile\ (#1)}
\algdef{SxnE}[IF]{If}{EndIf}[1]{\algorithmicif\ (#1)}
\algdef{cxnE}{IF}{Else}{EndIf}

\begin{document}


\titlebanner{banner above paper title}        
\preprintfooter{short description of paper}   

\title{Automated Feedback Generation for Introductory Programming Assignments}

\authorinfo{Rishabh Singh}
           {MIT CSAIL, Cambridge, MA}
           {rishabh@csail.mit.edu}
\authorinfo{Sumit Gulwani}
           {Microsoft Research, Redmond, WA}
           {sumitg@microsoft.com}
\authorinfo{Armando Solar-Lezama}
           {MIT CSAIL, Cambridge, MA}
           {asolar@csail.mit.edu}

\maketitle

\newtheorem{theorem}{Theorem}
\newtheorem{definition}{Definition}
\newtheorem{example}{Example}

\begin{abstract}

We present a new method for automatically providing feedback for
introductory programming problems. In order to use this method, we
need a reference implementation of the assignment, and an error model
consisting of potential corrections to errors that students might
make. Using this information, the system automatically derives minimal
corrections to student's incorrect solutions, providing them with a
quantifiable measure of exactly how incorrect a given solution was, as
well as feedback about what they did wrong.

We introduce a simple language for describing error models in terms of
correction rules, and formally define a rule-directed translation
strategy that reduces the problem of finding minimal corrections in an
incorrect program to the problem of synthesizing a correct program
from a sketch. We have evaluated our system on thousands of real
student attempts obtained from 6.00 and 6.00x. Our results show that
relatively simple error models can correct on average $65\%$ of
all incorrect submissions.

\end{abstract}







\section{Introduction}

There has been a lot of interest recently in making quality education
more accessible to students worldwide using information
technology. Several education initiatives such as EdX, Coursera, and
Udacity are teaming up with experts to provide online courses on
various college-level subjects ranging from computer science to
physics and psychology. These courses, also called massive open online
courses (MOOC), are typically taken by several thousands of students
worldwide, and presents many interesting challenges that are not
present in a traditional classroom setting consisting of only a few
hundred students. One such challenge in these courses is to provide
personalized feedback on practice exercises and assignments to a large
number of students. We consider the problem of providing automated
feedback for online introductory programming courses in this paper. We
envision this technology to be useful in a traditional classroom
setting as well.

The two most commonly used methods today for providing feedback on
programming problems are: (i) test-cases based feedback and (ii)
\emph{peer-feedback}~\cite{peerfeedback}. In automated test-cases
based feedback, the student program is run on a set of test cases and
the failing test cases are reported back as feedback to the
student. This is also how the 6.00x course (Introduction to Computer
Science and Programming) offered by MITx currently provides feedback
for the python programming exercises. The provided feedback of failing
test cases is however not ideal, especially for beginner programmers,
as they find it difficult to map the failing test cases to errors in
their code. We found a lot of students posting their submissions on
the discussion board seeking help from instructors and other students
after struggling for several hours to correct the mistakes
themselves. In fact, for the classroom version of the Introduction to
Programming course (6.00) taught at MIT, the teaching assistants are
required to manually go through each student submission and provide
qualitative feedback describing exactly what is wrong with the
submission and how to correct it. This manual feedback by teaching
assistants is simply prohibitive for the number of students in the
online class setting.

The second approach of peer-feedback is being suggested as a potential
solution to this problem. In this approach, the peer students who are
also taking the same course answer the posts on the discussion boards
-- this way the problem of providing feedback is distributed across
several peer students. Unfortunately, providing quality feedback is a
big challenge for experienced teaching assistants, and therefore it
presents an even bigger challenge for the peer students who are also
beginning to learn programming themselves. From the 6.00x discussion
boards, we observed that in many instances students had to wait
several hours (or days) to get any feedback, and in many cases the
feedback provided was either too general, incomplete or even wrong in
a few cases.

In this paper, we present an automated technique to provide feedback
for introductory programming assignments. The approach leverages
program synthesis technology to automatically determine minimal fixes
to the student's solution that will make it match the behavior of a
reference solution written by the instructor. This technology makes it
possible to provide students with precise feedback about what they did
wrong and how to correct them. The problem of providing automatic
feedback appears to be related to the problem of automated bug fixing,
but it differs from it in following two significant respects:
\begin{itemize}
\item \textbf{The complete specification is known.} An important
  challenge in automatic debugging is that there is no way to know
  whether a fix is addressing the root cause of a problem, or simply
  masking it and potentially introducing new errors. Usually the best
  one can do is check a candidate fix against a test suite or a
  partial specification~\cite{weimer}. While providing feedback on the
  other hand, the solution to the problem is known, and it is safe to
  assume that the instructor already wrote a correct reference
  implementation for the problem.

\item \textbf{Errors are predictable.} In a homework assignment,
  everyone is solving the same problem after having attended the same
  lectures, so errors tend to follow predictable patterns. This makes
  it possible to use a \emph{model-based} feedback approach, where the
  potential fixes are guided by a model of the kinds of errors
  students typically make for a given problem.

\end{itemize}
These simplifying assumptions, however, introduce their own set of
challenges. For example, since the complete specification is known,
the tool now needs to reason about the equivalence of the student
solution with the reference implementation. Also, in order to take
advantage of the predictability of errors, the tool needs to be
parameterized with models that describe the classes of errors. And
finally, these programs can be expected to have higher density of
errors than production code, so techniques like the one suggested by
\cite{rupakpldi11}, which attempts to correct bugs one path at a time
will not work for many of these problems that require coordinated
fixes in multiple places.

Our automated feedback generation technique handles all of these
challenges. The tool can reason about the semantic equivalence of
student programs and reference implementations written in a fairly
large subset of python, so the instructor does not have to learn a new
formalism to write specifications. The tool also provides an
\emph{error model} language that can be used to write an error model:
a very high level description of potential corrections to errors that
students might make in the solution. When the system encounters an
incorrect solution by a student, it symbolically explores the space of
all possible combinations of corrections allowed by the error model
and finds a correct solution requiring a \emph{minimal} set of
corrections.

We have evaluated our approach on thousands of student solutions on
programming problems obtained from the 6.00x submissions and
discussion boards, and from the 6.00 class submissions. These problems
constitute a major portion of first month of assignment problems. Our
tool can successfully provide feedback on over $65\%$ of the incorrect
solutions.

This paper makes the following key contributions:
\begin{itemize}
\item We show that the problem of providing automated feedback for
  introductory programming assignments can be framed as a synthesis
  problem. Our reduction uses a constraint-based mechanism to model
  python's dynamic typing and supports complex python constructs such
  as closures, higher-order functions and list comprehensions.
\item We define a high-level language \eml that can be used to provide
  correction rules to be used for providing feedback. We also show
  that a small set of such rules is sufficient to correct thousands of
  incorrect solutions written by students.
\item We report the successful evaluation of our technique on
  thousands of real student attempts obtained from 6.00 and 6.00x
  classes, as well as from \pexfun website. Our tool can provide
  feedback on 65\% of all submitted solutions that are incorrect in
  about $10$ seconds on average.
\end{itemize}

\section{Overview of the approach}
In order to illustrate the key ideas behind our approach, consider the
problem of computing the derivative of a polynomial whose coefficients
are represented as a list of integers. This problem is taken from week 3
problem set of 6.00x (PS3: Derivatives). Given the input list \t{poly},
the problem asks students to write the function \t{computeDeriv} that
computes a list \t{poly'} such that
{\small
$$
\t{poly'} = \left\{ \begin{array}{cl}
\{\t{i} \times \t{poly[i]} \where 0 < \t{i} < \t{len(poly)} \} & \mbox{if \t{len(poly)} > 1} \\
\lbrack 0.0 \rbrack & \mbox{if \t{len(poly)} = 1} 
\end{array} \right.
$$
}
For example, if the input list \t{poly} is $[2,-3,1,4]$ (denoting
$f(x) = 4x^3+x^2-3x+2$), the \t{computeDeriv} function should return
$[-3,2,12]$ (denoting the derivative $f'(x) = 12x^2+2x-3$). The
reference implementation for the \t{computeDeriv} function is shown in
\figref{compderivrefimpl}. This problem teaches concepts of
conditionals and iteration over lists. For this problem, students
struggled with many low-level python semantics issues such as the list
indexing and iteration bounds. In addition, they also struggled with
conceptual issues such as missing the corner case of handling lists
consisting of single element (denoting constant function).

\begin{figure}[!htpb]
\begin{lstlisting}[language=python]
def computeDeriv_list_int(poly_list_int):
    result = []
    for i in range(len(poly_list_int)):
        result += [i * poly_list_int[i]]
    if len(poly_list_int) == 1:
        return result      # return [0]
    else:
        return result[1:]  # remove the leading 0
\end{lstlisting}
\caption{The reference implementation for \t{computeDeriv} function.}
\figlabel{compderivrefimpl}
\end{figure}

\begin{figure*}[!htpb]
\begin{tabular}{c c}
\begin{minipage}{.4\linewidth}
\begin{center}
\begin{lstlisting}[language=python]
def computeDeriv(poly):
    deriv = []
    zero = 0
    if (len(poly) == 1):
        return deriv
    for expo in range (0, len(poly)):
        if (poly[expo] == 0):
            zero += 1
        else:
            deriv.append(poly[expo]*expo)
    return deriv
\end{lstlisting}
\end{center}
\end{minipage}

&
\begin{minipage}{.5\linewidth}
The program requires \textbf{3} changes:
\begin{itemize}
\item In the return statement \textbf{return deriv} in \textbf{line 5}, replace \textbf{deriv} by \textbf{[0]}.
\item In the comparison expression \textbf{(poly[expo] == 0)} in \textbf{line 7}, change \textbf{(poly[expo] == 0)} to \textbf{False}.
\item In the expression \textbf{range(0, len(poly))} in \textbf{line 6}, increment \textbf{0} by \textbf{1}.
\end{itemize}
\end{minipage}
\\
(a) Student's solution
&
(b) Generated Feedback
\\
\end{tabular}
\caption{(a) A student's \t{computeDeriv} solution from the 6.00x discussion board and (b) the feedback generated by our tool on this solution.}
\figlabel{compderivstudentsol}
\end{figure*}

A student solution for the \t{computeDeriv} problem taken from the
6.00x discussion
forum\footnote{\url{https://www.edx.org/courses/MITx/6.00x/2012_Fall/discussion/forum/600x_ps3_q2/threads/5085f3a27d1d422500000040}}
is shown in \figref{compderivstudentsol}(a). The student posted the
code in the forum seeking help and received two responses. The first
response asked the student to look for the first if-block return
value, and the second response said that the code should return $[0]$
instead of empty list for the first if statement. There are many
different ways to modify the code to return $[0]$ for the case
\t{len(poly)}=1. The student chose to change the initialization of the
\t{deriv} variable from $[~]$ to the list $[0]$. The problem with this
modification is that the result will now have an additional $0$ in
front of the output list for all input lists (which is undesirable for
lists of length greater than $1$). The student then posted the query
again on the forum on how to remove the leading $0$ from result, but
unfortunately this time did not get any more response.

Our tool generates the feedback shown in
\figref{compderivstudentsol}(b) for the student program in about $40$
seconds. During these $40$ seconds, the tool searches over more than
$10^{7}$ candidate fixes and finds the fix that requires minimum
number of corrections. There are three problems with the student code:
first it should return $[0]$ in line $5$ as was suggested in the forum
but wasn't specified how to make the change, second the if block
should be removed in line $7$, and third that the loop iteration
should start from index $1$ instead of $0$ in line $6$. The generated
feedback consists of four pieces of information (shown in bold in the
figure for emphasis):
\begin{itemize}
\item {the location of the error denoted by the line number.}
\item {the problematic expression in the line.}
\item {the sub-expression which needs to be modified.}
\item {the new modified value of the sub-expression.}
\end{itemize}

 The feedback generator is parameterized with a feedback-level
 parameter to generate feedback consisting of different combinations
 of the four information depending on how much information the
 instructor is willing to provide to the student.

\subsection{Workflow}
\label{sec:example-workflow}
In order to provide the level of feedback described above, the tool
needs some information from the instructor. First, the tool needs to
know what the problem is that the students are supposed to solve. The
instructor provides this information by writing a reference
implementation such as the one in \figref{compderivrefimpl}. Since
python is dynamically typed, the instructor also provides the types of
function arguments and return value. In \figref{compderivrefimpl}, the
instructor specifies the type of input argument to be list of integers
(\t{poly\_list\_int}) by appending the type to the name.

In addition to the reference implementation, the tool needs a
description of the kinds of errors students might make. We have
designed an error model language $\eml$, which can describe a set of
correction rules that denote the potential corrections to errors that
students might make. For example, in the student attempt in
\figref{compderivstudentsol}(a), we observe that corrections often
involve modifying the return value and the range iteration values. We
can specify this information with the following three correction
rules:
\begin{eqnarray*}
\t{return} ~ \aexpr & \rightarrow & \t{return} ~ [0]\\
\t{range}(\aexpr_1,\aexpr_2) & \rightarrow & \t{range}(\aexpr_1+1,\aexpr_2)\\
\aexpr_0 == \aexpr_1 & \rightarrow & \t{False}
\end{eqnarray*}
The correction rule $\t{return} ~ \aexpr \rightarrow \t{return} ~ [0]$
states that the expression of a \t{return} statement can be optionally
replaced by $[0]$. The error model for this problem that we use for
our experiments is shown in \figref{compderiverrormodel}, but we will
use this simple error model for simplifying the presentation in this
section. In later experiments, we also show how only a few tens of
incorrect solutions can provide enough information to create an error
model that can automatically provide feedback for thousands of
incorrect solutions.

The tool now needs to explore the space of all candidate programs
based on applying these correction rules to the student program, and
compute the candidate program that is equivalent to the reference
implementation and that requires minimum number of corrections. We use
constraint-based synthesis technology~\cite{sketchthesis, GulwCBA,
  SrivPts} to efficiently search over this large space of
programs. Specifically, we use the \Sk{} synthesizer that uses a
sat-based algorithm to complete program sketches (programs with holes)
so that they meet a given specification. We extend the \Sk{}
synthesizer with support for \emph{minimize} hole expressions whose
values are computed efficiently by using incremental constraint
solving. To simplify the presentation, we use a simpler language
$\impl$ (\emph{miniPython}) in place of python to explain the details
of our algorithm. In practice, our tool supports a fairly large subset
of python including closures, higher order functions and list
comprehensions.

\subsection{Solution Strategy}

\begin{figure}[!htpb]
\begin{center}
\includegraphics[scale=0.4]{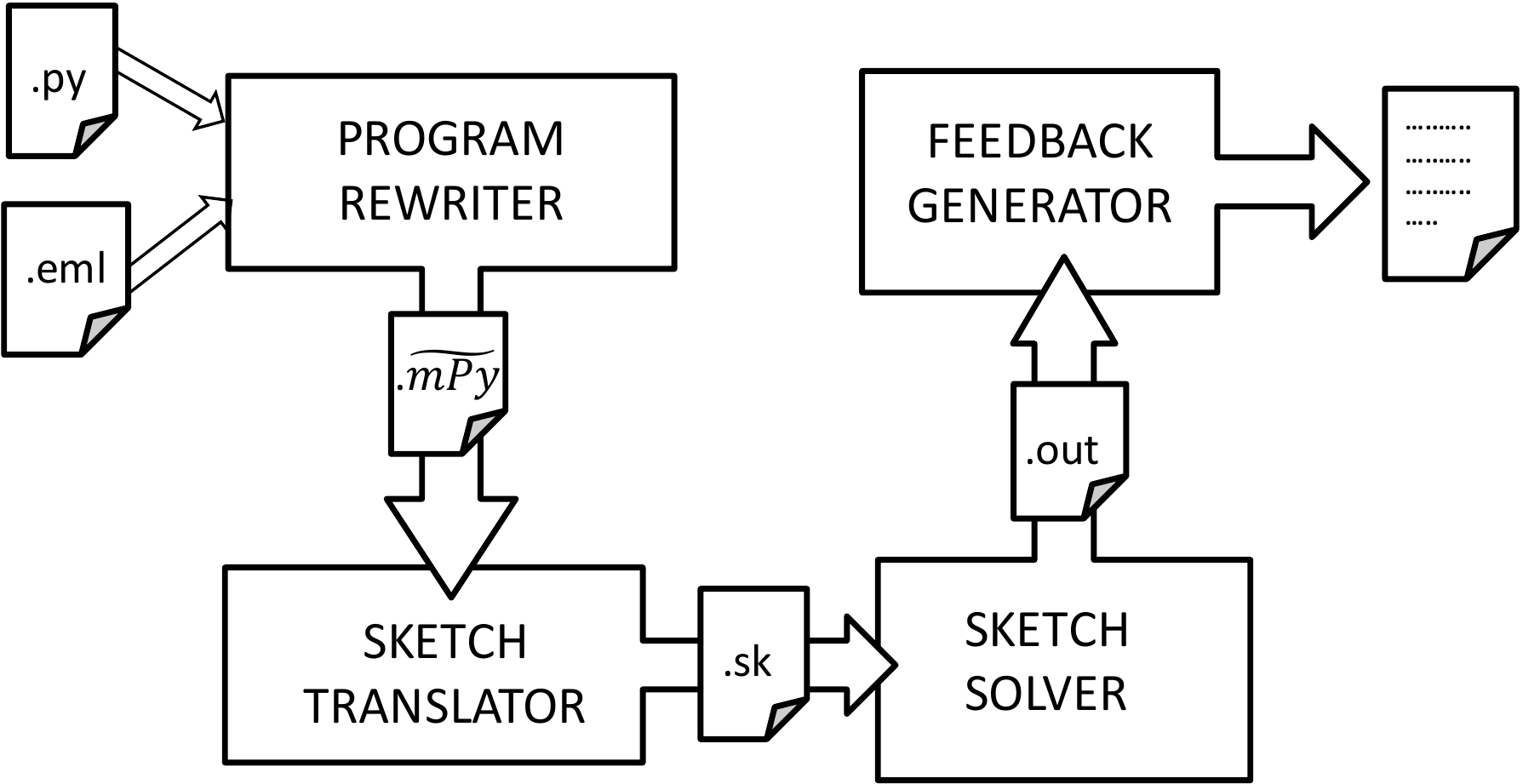}
\end{center}
\caption{The architecture of our automated feedback generation tool.}
\figlabel{tool-arch}
\end{figure}

The architecture of our tool is shown in \figref{tool-arch}. The
solution strategy to find minimal corrections to a student's solution
is based on a two-phase translation to the Sketch synthesis
language. In the first phase, the \t{Program Rewriter} uses the
correction rules to translate the solution into a language we call
$\implp$; this language provides us with a concise notation to
describe sets of $\impl$ candidate programs, together with a cost
model to reflect the number of corrections associated with each
program in this set. In the second phase, this $\implp$ program is
translated into a sketch program by the \t{Sketch Translator}.

In the case of example in \figref{compderivstudentsol}(a), the
\t{Program Rewriter} produces the $\implp$ program shown in
\figref{fig-implp-ex} using the correction rules from
\secref{example-workflow}. This program includes all the possible
corrections induced by the correction rules in the model. The $\implp$
language extends the imperative language $\impl$ with expression
choices, where the choices are denoted with squiggly
brackets. Whenever there are multiple choices for an expression or a
statement, the zero-cost choice, the one that will leave the
expression unchanged, is boxed. For example, the expression choice
$\{\defchoice{\aexpr_0}, \aexpr_1, \cdots, \aexpr_n\}$ denotes a
choice between expressions $\aexpr_0$, $\cdots$, $\aexpr_n$ where
$\aexpr_0$ denotes the zero-cost default choice.

\begin{figure}
\begin{lstlisting}[language=python,mathescape=true]
def computeDeriv(poly):
    deriv = []
    zero = 0
    if ({$\defchoice{\t{len(poly) == 1}}$, False}):
        return {$\defchoice{\t{deriv}}$,[0]}
    for expo in range ({$\defchoice{\t{0}}$,1}, len(poly)):
        if ({$\defchoice{\t{poly[expo] == 0}}$,False}):
            zero += 1
        else:
            deriv.append(poly[expo]*expo)
    return {$\defchoice{\t{deriv}}$,[0]}
\end{lstlisting}
\caption{The resulting $\implp$ program after applying correction
  rules to program in \figref{compderivstudentsol}(a).}
\figlabel{fig-implp-ex}
\end{figure}

For this simple program, the three correction rules induce a space of
$32$ different candidate programs. This candidate space is fairly
small, but the number of candidate programs grow exponentially with
the number of correction places in the program and with the number of
correction choices in the rules. The error model that we use in our
experiments induces a space of more than $10^{12}$ candidate programs
for some of the benchmark problems. In order to search this large
space efficiently, the program is translated to a sketch by the
\t{Sketch Translator}.

\subsection{Synthesizing Corrections with Sketch}
\seclabel{overview-translation}
The \Sk{}~\cite{sketchthesis} synthesis system allows programmers to
write programs while leaving fragments of it unspecified as
\emph{holes}; the contents of these holes are filled up automatically
by the synthesizer such that the program conforms to a specification
provided in terms of a reference implementation. The synthesizer uses
the \t{CEGIS} algorithm~\cite{solar2005bitstreaming} to efficiently
compute the values for holes and uses bounded symbolic verification
techniques for performing equivalence check of the two implementations.

There are two key aspects in the translation of an $\implp$ program to
a \Sk{} program. The first aspect is specific to the python
language. \Sk{} supports high-level features such as closures and
higher-order functions which simplifies the translation, but it is
statically typed whereas $\impl$ programs (like python) are
dynamically typed. The translation models the dynamically typed
variables and operations over them using struct types in \Sk{} in a
way similar to the union types. The second aspect of the translation
is the modeling of set-expressions in $\implp$ using \t{??} (holes) in
\Sk{}, which is language independent.

\begin{figure}
\begin{tabular}{c c}
{\begin{lstlisting}[language=C, numbers=none]
struct MultiType{
  int val, type;
  MTList lst;
}
\end{lstlisting}}
&~~
{\begin{lstlisting}[language=C, numbers=none]
struct MTList{
  int len;
  MultiType[len] lVals;
}
\end{lstlisting}}
\end{tabular}
\caption{The definition of \t{MultiType} struct for encoding dynamic types in python.}
\figlabel{multitypestruct}
\end{figure}
The dynamic variable types in $\impl$ language are modeled using the
\t{MultiType} struct defined in \figref{multitypestruct}. The
\t{MultiType} struct consists of a \t{type} field that denotes the
dynamic type of variables and currently supports the following set of
types \{\t{INTEGER}, \t{BOOL}, \t{TYPE}, \t{LIST}, \t{TUPLE}\}. The
\t{val} field stores the integer value or the Boolean value of the
variables, whereas the \t{lst} field of type \t{MTList} stores the
value of list and tuple variables. The \t{MTList} struct consists of a
field \t{len} that denotes the length of the list and a field
\t{lVals} of type array of \t{MultiType} that stores the list
elements. For example, the integer value \t{5} is represented as the
value \t{MultiType(val=5, flag=INTEGER)} and the list \t{[1,2]} is
represented as the value \t{MultiType(lst=new MTList(len=2,lVals=\{new
  MultiType(val=1,flag=INTEGER), new MultiType(val=2,flag=INTEGER)\}),
  flag=LIST)}.

\techrep{
\begin{figure}[!htpb]
\begin{lstlisting}[language=C, morekeywords={implements, bit, MultiType, minimize}]
  #include <multiTypeLib.c>
  int totalCost = 0;
  bit choiceComp0=0, choiceComp1=0;
  bit choiceIdx0=0, choiceRetVal0=0, choiceRetVal1=0;

  #define zeroMT new MultiType(val=0, type = INTEGER)
  #define oneMT new MultiType(val=1, type = INTEGER)
  #define emptyListMT = new MultiType(lst = new MTList(lVals={}, len=0), type = LIST)

  MultiType computeDeriv(MultiType poly){
    MultiType deriv, zero;
    assignMT(deriv, emptyListMT);
    assignMT(zero, zeroMT);
    if(getBoolValue(modComp0(compareMT(len(poly), oneMT, COMP_EQ))))
      return modRetVal0(deriv);
    MultiType iterList = range2(modIdx0(zeroMT),len(poly));
    for(int i=0; i< iterList.lst.len; i++){
      MultiType expo = iterList.lst.lVals[i];
      if(getBoolValue(modComp1(compareMT(sub(poly,expo), zeroMT, COMP_EQ))))
        assignMT(zero, binOpMT(zero, oneMT, ADD_OP));
      else
        append(deriv, binOpMT(sub(poly,expo),expo, MUL_OP));
    }
    return modRetVal1(deriv);
  }
  
  MultiType modComp0(MultiType b){
    if(??) return b;
      choiceComp0 = 1;
      return FalseMT();
  }
  MultiType modIdx0(MultiType a){
    if(??) return a;
      choiceIdx0 = 1;
      return incrementOne(a);
  }
  MultiType modRetVal0(MultiType a){
    if(??) return a;
    choiceRetVal0 = 1;
    MTList list = new MTList(lVals={zeroMT}, len =1);
    return new MultiType(lst=list, type = LIST);
  }

  int[N] computeDeriv_driver(int N, int[N] poly) implements computeDeriv_teacher_driver{
    MultiType polyMT = createMTFromArray(N, poly);
    MultiType resultMT = computeDeriv(polyMT);
    int[N] result = getArrayFromMT(resultMT);
    if(choiceComp0) totalCost = totalCost + 1;
    if(choiceComp1) totalCost = totalCost + 1;
    if(choiceIdx0) totalCost = totalCost + 1;
    if(choiceRetVal0) totalCost = totalCost + 1;
    if(choiceRetVal1) totalCost = totalCost + 1;
    minimize(totalCost);
    return result;
  }

\end{lstlisting}
\caption{The sketch generated for $\implp$ program in \figref{fig-implp-ex}.}
\figlabel{fig-sketch-ex}
\end{figure}
}

 The second key aspect of this translation is the translation of
 expression choices in $\implp$. The \Sk{} construct \t{??} denotes an
 unknown integer hole that can be assigned any constant integer value
 by the synthesizer. The expression choices in $\implp$ are translated
 to functions in \Sk{} that based on the unknown hole values return
 either the default expression or one of the other expression
 choices. Each such function is associated with a unique Boolean
 choice variable, which is set by the function whenever it returns a
 non-default expression choice. For example, the set-statement
 \t{return} \{$\defchoice{\t{deriv}}$,[0]\}; (line $5$ in
 \figref{fig-implp-ex}) is translated to \t{return modRetVal0(deriv)},
 where the \t{modRetVal0} function is defined as:

\begin{lstlisting}[language=C, numbers=none]
  MultiType modRetVal0(MultiType a){
    if(??) return a; // default choice
    choiceRetVal0 = True; // non-default choice
    MTList list = new MTList(lVals={new MultiType(val=0, flag=INTEGER)}, len=1);
    return new MultiType(lst=list, type = LIST);
  }
\end{lstlisting}

The translation phase also generate a \Sk{} harness that calls and
compares the outputs of the translated student and reference
implementations on all inputs of a bounded size. For example in case
of the \t{computeDeriv} function that takes a list as input, with
bounds of $n=4$ for both the number of integer bits and the maximum
length of the input list, the harness makes sure that the output of
the two implementations match for more than $2^{16}$ different inputs
as opposed to $10$ test-cases used in 6.00x. The harness also defines
a \t{totalCost} variable as a function of all choice variables that
computes the total number of modifications performed in the original
program, and asserts that the value of \t{totalCost} should be
minimized. The synthesizer then solves this minimization problem
efficiently using an incremental solving algorithm \t{CEGISMIN}
described in \secref{alg-increment-solving}.

After the synthesizer finds a solution, the \t{Feedback Generator}
uses the solution to the unknown integer holes in the sketch to
compute the choices made by the synthesizer and generates the
corresponding feedback. For this example, the tool generates the
feedback shown in \figref{compderivstudentsol}(b) in less than $40$
seconds.


\section{\eml: Error Model Language}
\seclabel{fuzzlanguage} In this section, we describe the syntax and
semantics of the error model language $\eml$. An $\eml$ error model
consists of a set of rewrite rules that captures the potential
corrections for mistakes that students might make in their
solutions. We define the rewrite rules over a simple python-like
imperative language $\impl$. A rewrite rule transforms a program
element in $\impl$ to a set of weighted $\impl$ program elements. This
weighted set of $\impl$ program elements is represented succinctly as
an $\implp$ program element, where $\implp$ extends the $\impl$
language with set-exprs (set of expressions) and set-stmts (set of
statements). The weight associated with a program element in this set
denotes the cost of performing the corresponding correction. An error
model transforms an $\impl$ program to an $\implp$ program
(representing a set of $\impl$ programs) by recursively applying the
rewrite rules. We show that this transformation is deterministic and
is guaranteed to terminate on \emph{well-formed} error models.

\subsection{$\impl$ and $\implp$ languages}

\begin{figure*}[!htpb]
\begin{tabular}{c|c}
\begin{minipage}{3in}
\begin{eqnarray*}
\mbox{Arith Expr } \aexpr & \bnfas & \intN ~ \bnfalt [~]  ~ \bnfalt ~ \var ~ \bnfalt ~ \aexpr[\aexpr] ~ \bnfalt ~ \aexpr_0 ~ \aop ~ \aexpr_1\\
& \bnfalt & [\aexpr_1, \cdots, \aexpr_n] ~ \bnfalt ~ \fvar(\aexpr_0, \cdots, \aexpr_n)\\
& \bnfalt & \aexpr_0 ~ \t{if} ~ \bexpr ~ \t{else} ~ \aexpr_1\\
\mbox{Arith Op } \aop & \bnfas & + ~ \bnfalt ~ - ~ \bnfalt ~ \times ~ \bnfalt ~/ ~ \bnfalt **\\
\mbox{Bool Expr } \bexpr & \bnfas & \t{not} ~ \bexpr ~ \bnfalt ~ \aexpr_0 ~ \cop ~ \aexpr_1 ~ \bnfalt ~\bexpr_0 ~ \bop ~ \bexpr_1\\
\mbox{Comp Op } \cop & \bnfas & == ~ \bnfalt ~ < ~ \bnfalt  ~ > ~ \bnfalt~ \leq ~ \bnfalt \geq\\
\mbox{Bool Op } \bop & \bnfas & \t{and} ~ \bnfalt ~ \t{or}\\
\mbox{Stmt Expr } \sexpr & \bnfas & \var = \aexpr ~ \bnfalt ~ \sexpr_0;\sexpr_1 ~ \bnfalt ~ \t{while} ~ \bexpr: ~\sexpr\\
 &\bnfalt & \t{if} ~ \bexpr: ~ \sexpr_0 ~ \t{else:} ~ \sexpr_1\\
& \bnfalt & \t{for}~\aexpr_0~\t{in}~\aexpr_1:~\sexpr ~ \bnfalt ~ \t{return} ~ \aexpr\\
\mbox{Func Def. } \prog & \bnfas & \t{def} ~ f(\aexpr_1, \cdots, \aexpr_n): ~\sexpr
\end{eqnarray*}
\end{minipage}
&
\begin{minipage}{3.5in}
\begin{eqnarray*}
\mbox{Arith set-expr } \saexpr & \bnfas & \aexpr ~ \bnfalt ~ \{\defchoice{\saexpr_0},\cdots, \saexpr_n\} ~ \bnfalt ~ \saexpr[\saexpr] ~ \bnfalt ~ \saexpr_0 ~ \saop ~ \saexpr_1\\
& \bnfalt & [\saexpr_0,\cdots,\saexpr_n] ~ \bnfalt ~ \fs(\saexpr_0, \cdots, \saexpr_n)\\
\mbox{set-op } \sop_x & \bnfas & \aop ~ \bnfalt \{\defchoice{\sop_{x_0}},\cdots, \sop_{x_n} \}\\
\mbox{Bool set-expr } \sbexpr & \bnfas & \bexpr ~ \bnfalt ~ \{\defchoice{\sbexpr_0},\cdots, \sbexpr_n\} ~ \bnfalt ~\t{not} ~ \sbexpr ~ \bnfalt ~ \saexpr_0 ~ \scop ~ \saexpr_1 ~ \bnfalt ~\sbexpr_0 ~ \sbop ~ \sbexpr_1\\
\mbox{Stmt set-expr } \ssexpr & \bnfas & \sexpr ~ \bnfalt ~ \{\defchoice{\ssexpr_0},\cdots, \ssexpr_n\} ~ \bnfalt ~ \svar := \saexpr ~ \bnfalt ~ \ssexpr_0;\ssexpr_1 \\
& \bnfalt & \t{while} ~ \sbexpr: ~ \ssexpr ~ \bnfalt ~ \t{for}~\saexpr_0~\t{in}~\saexpr_1:~\ssexpr\\
 &\bnfalt & \t{if} ~ \sbexpr: ~ \ssexpr_0 ~ \t{else}: ~ \ssexpr_1 ~ \bnfalt ~ \t{return} ~ \saexpr\\
\mbox{Func Def } \sprog & \bnfas & \t{def} ~ f(\aexpr_1, \cdots, \aexpr_n) ~\ssexpr
\end{eqnarray*}
\end{minipage}\\\\
(a)~$\impl$ & (b)~$\implp$ \\
\end{tabular}
\caption{The syntax for (a) $\impl$ and (b) $\implp$ languages.}
\label{fig-lang-table}
\end{figure*}

The syntax for the simple imperative language $\impl$ is shown in
Figure~\ref{fig-lang-table}(a) and the syntax of $\implp$ language is
shown in Figure~\ref{fig-lang-table}(b). The purpose of $\implp$
language is to represent a large collection of $\impl$ programs
succinctly. The $\implp$ language consists of set-expressions
($\saexpr$ and $\sbexpr$) and set-statements ($\ssexpr$) that
represent a weighted set of corresponding $\impl$ expressions and
statements respectively. For example, the set expression
$\{\defchoice{n_0},\cdots,n_k\}$ represents a weighted set of constant
integers where $n_0$ denotes the default integer value associated with
cost $0$ and all other integer constants ($n_1,\cdots,n_k$) are
associated with cost $1$. The sets of composite expressions are
represented succinctly in terms of sets of their constituent
sub-expressions. For example, the composite expression
$\{\defchoice{a_0},a_0+1\}\{\defchoice{<},\le,>,\ge,==,\neq\}\{\defchoice{a_1},a_1+1,a_1-1\}$
represents $36$ $\impl$ expressions.

Each $\impl$ program in the set of programs represented by an $\implp$
program is associated with a cost (weight) that encodes the number of
modifications performed in the original program to obtain the
transformed program. This cost allows the tool to search for
corrections that require minimum number of modifications. The weighted
set of $\impl$ programs is defined using the $\sem{~}$ function shown
\nottechrep{partially} in \figref{fig-implp-cost}\nottechrep{, the
  complete function definition can be found in \cite{techrep}}. The
$\sem{~}$ function on $\impl$ expressions such as $\aexpr$ returns a
singleton set $\{(a,0)\}$ consisting of the corresponding expression
that is associated with cost 0. On set-expressions of the form
$\{\defchoice{\saexpr_0}, \cdots, \saexpr_n\}$, the function returns
the union of the weighted set of $\impl$ expressions corresponding to
the default set-expression ($\sem{\saexpr_0}$) and the weighted set of
expressions corresponding to other set-expressions
($\saexpr_1,\cdots,\saexpr_n$), where each expression in
$\sem{\saexpr_i)}$ is associated with an additional cost of $1$. On
composite expressions, the function computes the weighted set
recursively by taking the cross-product of weighted sets of its
constituent sub-expressions and adding their corresponding costs. For
example, the weighted set for composite expression
$\tilde{x}[\tilde{y}]$ consists of an expression $x_i[y_j]$ associated
with cost $c_{x_i}+c_{y_j}$ for each $(x_i,c_{x_i}) \in
\sem{\tilde{x}}$ and $(y_j,c_{y_j})\in \sem{\tilde{y}}$.

\nottechrep{
\begin{figure}[!htpb]

\begin{eqnarray*}
\sem{\aexpr} & = &  \{(\aexpr,0)\}\\
\sem{\{\defchoice{\saexpr_0}, \cdots, \saexpr_n\}} & = &  \sem{\saexpr_0} \cup \{(\aexpr,~c+1) \where (\aexpr,c)\in\sem{\saexpr_i}_{0 < i \leq n} \}\\
\sem{\saexpr_0[\saexpr_1]} & = & \{(\aexpr_0[\aexpr_1],~c_0+c_1) \where (\aexpr_i,c_i)\in \sem{\saexpr_i}_{i \in \{0,1\}}\}\\
\sem{\t{while} ~\sbexpr: \ssexpr} & = & \{(\t{while} ~ \bexpr: ~ \sexpr, ~ c_b + c_s) \where\\
& & \hspace*{15mm}(\bexpr,c_b) \in \sem{\sbexpr}, (\sexpr,c_s) \in \sem{\ssexpr}\}
\end{eqnarray*}

\caption{The $\sem{~}$ function (shown partially) that translates an
  $\implp$ program to a weighted set of $\impl$ programs.}
\figlabel{fig-implp-cost}
\end{figure}
}

\techrep{
\begin{figure*}[!htpb]
\begin{eqnarray*}
\sem{\aexpr} & = &  \{(\aexpr,0)\}\\
\sem{\{\defchoice{\saexpr_0}, \cdots, \saexpr_n\}} & = &  \sem{\saexpr_0} \cup \{(\aexpr,~c+1) \where (\aexpr,c)\in\sem{\saexpr_i}, 0 < i \leq n \}\\
\sem{\saexpr_0[\saexpr_1]} & = & \{(\aexpr_0[\aexpr_1],~c_0+c_1) \where (\aexpr_0,c_0)\in \sem{\saexpr_0}, (\aexpr_1,c_1) \in \sem{\saexpr_1} \}\\
\sem{\saexpr_0 ~ \sxop ~ \saexpr_1} & = & \{(\aexpr_0 ~ \xop ~ \aexpr_1, ~c_0+c_x+c_1) \where (\aexpr_0,c_0)\in\sem{\saexpr_0}, (\xop,c_x)\in\sem{\sxop}, (\aexpr_1,c_1)\in\sem{\saexpr_1}, x \in \{a,b,c\} \}\\
\sem{ \fs(\saexpr_1, \cdots, \saexpr_n)} & = & \{(f(\aexpr_1,\cdots,\aexpr_n),~c_f+c_1\cdots+c_n) \where (f,c_f)\in\sem{\fs}, (\aexpr_i,c_i) \in \sem{\saexpr_i} \}\\
\sem{\aop} & = &  \{(\aop,0)\}\\
\sem{\{\defchoice{\sop_{a_0}}, \cdots, \sop_{a_n}\}} & = &  \sem{\sop_{a_0}} \cup \{(\aop,~c+1) \where (\aop,c)\in\sem{\sop_{a_i}} \}\\
\sem{\bexpr} & = &  \{(\bexpr,0)\}\\
\sem{\{\defchoice{\sbexpr_0}, \cdots, \sbexpr_n\}} & = & \sem{\sbexpr_0} \cup \{(\bexpr,~c+1) \where (\bexpr,c)\in\sem{\sbexpr_i}, 0 < i \leq n \}\\
\sem{\sexpr} & = &  \{(\sexpr,0)\}\\
\sem{\{\defchoice{\ssexpr_0}, \cdots, \ssexpr_n\}} & = & \sem{\ssexpr_0} \cup \{(\sexpr,~c+1) \where (\sexpr,c)\in\sem{\ssexpr_i}, 0 < i \leq n \}\\
\sem{\svar := \saexpr} & = &  \{(\var := \aexpr, ~c_0 + c_1) \where (\var,c_0) \in \sem{\svar}, (\aexpr,c_1) \in \sem{\saexpr}\}\\
\sem{\ssexpr_0;\ssexpr_1} & = & \{ (\sexpr_0;\sexpr_1, ~c_0+c_1) \where (\sexpr_0,c_0) \in \sem{\ssexpr_0}, (\sexpr_1,c_1) \in \sem{\ssexpr_1}\}\\
\sem{\t{if} ~ \sbexpr ~ \t{then} ~ \ssexpr_0 ~ \t{else} ~ \ssexpr_1} & = & \{(\t{if} ~ \bexpr ~ \t{then} ~ \sexpr_0 ~ \t{else} ~ \sexpr_1, ~c_b+c_0+c_1) \where (\bexpr,c_b) \in \sem{\sbexpr}, (\sexpr_0,c_0) \in \sem{\ssexpr_0}, (\sexpr_1,c_1) \in \sem{\ssexpr_1}\}\\
\sem{\t{while} ~ \sbexpr ~ \t{do} ~ \ssexpr} & = & \{(\t{while} ~ \bexpr ~ \t{do} ~ \sexpr, ~ c_b + c_s) \where (\bexpr,c_b) \in \sem{\sbexpr}, (\sexpr,c_s) \in \sem{\ssexpr}\}\\ 
\sem{\t{return} ~ \saexpr} & = & \{(\t{return} ~ \aexpr, ~ c) \where (\aexpr,c) \in \sem{\aexpr}\}\\ 
\end{eqnarray*}
\caption{The $\sem{~}$ function that translates an
  $\implp$ program to a weighted set of $\impl$ programs.}
\figlabel{fig-implp-cost}
\end{figure*}
}

\subsection{Syntax of $\eml$}

An $\eml$ error model consists of a set of correction rules that are
used to transform an $\impl$ program to an $\implp$ program. A {\em
  correction rule} $\crule$ is written as a rewrite rule $L
\rightarrow R$, where $L$ and $R$ denote a \emph{program element} in
$\impl$ and $\implp$ respectively. A program element can either be a
term, an expression, a statement, a method or the program itself. The
left hand side ($L$) denotes an $\impl$ program element that is
pattern matched to be transformed to an $\implp$ program element
denoted by the right hand side ($R$). The left hand side of the rule
can use free variables whereas the right hand side can only refer to
the variables present in the left hand side. The language also
supports a special $\prime$ (prime) operator that can be used to
\emph{tag} sub-expressions in $R$ that are further transformed
recursively using the error model. The rules use a shorthand notation
$\choosevar \aexpr$ (in the right hand side) to denote the set of all
variables that are of the same type as the type of expression $\aexpr$
and are in scope at the corresponding program location. We assume each
correction rule is associated with cost $1$, but it can be easily
extended to different costs to account for different levels of
mistakes.

\begin{figure}[!htpb]
\begin{eqnarray*}
\mbox{\indF:} ~~\var[\aexpr] & \rightarrow & \var[\{\aexpr+1, \aexpr -
  1, \choosevar\aexpr\}]\\ \mbox{\initF:}
~~\var=\intN & \rightarrow & \var = \{n+1, n-1, 0\}\\ 
\mbox{\rangeF:} ~~ \t{range}(\aexpr_0,\aexpr_1) &  \rightarrow & \t{range}(\{0,1,\aexpr_0-1,\aexpr_0+1\},\\
& &  \hspace*{10mm}\{\aexpr_1+1,\aexpr_1-1\})\\
\mbox{\condF:}
~~ \aexpr_0 ~ \cop ~ \aexpr_1 ~ & \rightarrow & \{\{\aexpr_0\dash-1, ?\aexpr_0 \} ~\scop~\{\aexpr_1\dash-1, 0, 1, ?\aexpr_1 \}, \\ 
& & \hspace*{2mm}\t{True}, \t{False}\}\\
&
& \mbox{where } \scop = \{<, >, \leq, \geq, ==,
\neq\}\\ 
\mbox{\retF:} ~~ \t{return} ~\aexpr & \rightarrow & \t{return} \{[0]~\t{if}~\t{len}(\aexpr)==1 ~ \t{else}~\aexpr,\\
& & \hspace*{6mm} \aexpr[1:] ~ \t{if}~(\t{len}(\aexpr) > 1) ~\t{else} ~\aexpr\} \end{eqnarray*}
\caption{The error model $\fm$ for the \t{computeDeriv} problem.}
\figlabel{compderiverrormodel}
\end{figure}

\begin{example}
\label{ex-array-sort-model}
The error model for the \t{computeDeriv} problem is shown in
\figref{compderiverrormodel}. The \indF rewrite rule transforms the
list access indices. The \initF rule transforms the right hand size of
constant initializations. The \rangeF rule transforms the arguments
for the \t{range} function; similar rules are defined in the model for
other \t{range} functions that take one and three arguments. The
\condF rule transforms the operands and operator of the
comparisons. The \retF rule adds the two common corner cases of
returning $[0]$ when the length of input list is $1$, and the case of
deleting the first list element before returning the list. Note that
these rewrite rules define the corrections that can be performed
optionally; the zero cost (default) case of not correcting a program
element is added automatically as described in \secref{eml-semantics}.

\end{example}

\begin{definition} Well-formed Rewrite Rule : A rewrite rule $\crule : L \rightarrow R$ is defined to be well-formed if all tagged sub-terms $t\dash$ in $R$ have a smaller size syntax tree than that of $L$.
\end{definition}

The rewrite rule $\crule_1: \var[\aexpr] \rightarrow
\{(\var[\aexpr])\dash + 1\}$ is not a well-formed rewrite rule as the
size of the tagged sub-term ($\var[\aexpr]$) of $R$ is the same as
that of the left hand side $L$. On the other hand, the rewrite rule
$\crule_2: \var[\aexpr] \rightarrow \{\var\dash[\aexpr\dash]+1\}$ is
well-formed.

\begin{definition} Well-formed Error Model : An error model $\fm$ is defined to be well-formed if all of its constituent rewrite rules $\crule_i \in \fm$ are well-formed.
\end{definition} 

\subsection{Transformation with $\eml$}
\seclabel{eml-semantics} An error model $\emod$ is syntactically
translated to a function $\transExpr_\emod$ that transforms an $\impl$
program to an $\implp$ program. The $\transExpr_\emod$ function first
traverses the program element $\wexpr$ in the default way, i.e. no
transformation happens at this level of the syntax tree, and the
method is called recursively on all of its top-level sub-terms $t$ to
obtain the transformed element $\wexpr_0 \in \implp$. For each
correction rule $\crule_i : L_i \rightarrow R_i$ in the error model
$\emod$, the method contains a \matchE expression that matches the
term $\wexpr$ with the left hand side of the rule $L_i$ (with
appropriate unification of the free variables in $L_i$). If the match
succeeds, it is transformed to a term $\wexpr_i \in \implp$ as defined
by the right hand side $R_i$ of the rule after applying the
$\transExpr_\emod$ method recursively on each one of its tagged
sub-terms $t\dash$. Finally, the method returns the set of all
transformed terms $\{\defchoice{\wexpr_0},\cdots,\wexpr_n \}$.

\techrep{
\begin{figure}[!htpb]
\begin{eqnarray*}
\transExpr_\fm(\wexpr: \impl) : \implp & =\hspace*{20mm}
& \\ & & \hspace*{-60mm} \t{let} ~ \wexpr_0 =
\wexpr[t\rightarrow\transExpr_\fm(t)] ~\t{in}\\ 
& & \hspace*{-50mm} \cdots\cdots\\
& & \hspace*{-60mm} \t{let} ~ \wexpr_i = \matchE ~ \wexpr ~\t{with}\\
 & & \hspace*{-47mm} ~L_i ~ \rightarrow ~ R_i[t'~\rightarrow~\transExpr_\fm(t)] ~ \t{in}\\
& & \hspace*{-50mm} \cdots\\ & & \hspace*{-60mm}
\{\defchoice{\wexpr_0},\cdots,\wexpr_n \}
\end{eqnarray*}
\caption{The syntactic translation of an error model $\emod$ to
  $\transExpr_\emod$ function.}  \figlabel{transsemantics}
\end{figure}
}

\begin{example}
\label{ex-emapplication}
Consider an error model $\emod_1$ consisting of the following three correction rules:
\begin{eqnarray*}
\mbox{$\crule_1 : $} ~~ \var[\aexpr] & \rightarrow & \var[\{\aexpr-1, \aexpr+1\}]\\
\mbox{$\crule_2 : $}~~ \aexpr_0 ~ \cop ~  \aexpr_1 ~ & \rightarrow & \{\aexpr_0\dash -1, 0 \} ~ \cop ~ \{\aexpr_1\dash -1, 0 \} \\ 
\mbox{$\crule_3 : $} ~~ \var[\aexpr] & \rightarrow & \choosevar\var[\aexpr]\\
\end{eqnarray*}

The transformation function $\transExpr_{\emod_1}$ for the error model
$\emod_1$ is shown in \figref{extranssem}.
\end{example}

\begin{figure}[!htpb]
\begin{eqnarray*}
\transExpr_{\emod_1}(\wexpr: \impl) : \implp & =  \hspace*{20mm}&\\
& & \hspace*{-60mm} \t{\bf let} ~ \wexpr_0 = \wexpr[t\rightarrow\transExpr_{\emod_1}(t)] ~ \t{\bf in} ~~ (*~~\mbox{$t$ : a sub-term of $w$}~~ *)\\
& & \hspace*{-60mm} \t{\bf let} ~ \wexpr_1 = \t{\bf Match} ~ \wexpr ~\t{\bf with} \\
& & \hspace*{-45mm} \var[\aexpr] \rightarrow \var[\{\aexpr+1, \aexpr-1\}] ~ \t{\bf in}\\
& & \hspace*{-60mm} \t{\bf let} ~  \wexpr_2 = \t{\bf Match} ~\wexpr ~\t{\bf with} \\
& & \hspace*{-45mm} \aexpr_0 ~ \cop ~\aexpr_1 \rightarrow \{\transExpr_{\emod_1}(\aexpr_0)-1,0\} ~ \cop\\
& & \hspace*{-20mm} \{\transExpr_{\emod_1}(\aexpr_1)-1,0\} ~\t{\bf in}\\
& & \hspace*{-60mm} \{\defchoice{\wexpr_0},\wexpr_1,\wexpr_2\}
\end{eqnarray*}
\caption{The $\transExpr_{\emod_1}$ method for error model $\emod_1$.}
\figlabel{extranssem}
\end{figure}

\begin{figure*}
$
\begin{array}{c}
\tr(x[i]<y[j]) \equiv \{\defchoice{\tr(x[i])<\tr(y[j])}, \{\tr(x[i])-1, 0\}< \{\tr(y[j])-1, 0\}\}
\\\\
\tr(x[i]) \equiv \{\defchoice{\tr(x)[\tr(i)]}, x[\{i+1,i-1\}], y[i]\}\\\\
\tr(y[j]) \equiv \{\defchoice{\tr(y)[\tr(j)]}, y[\{j+1,j-1\}], x[j]\}\\\\
\tr(x) \equiv \{\defchoice{x}\} \hspace*{10mm} \tr(i) \equiv \{\defchoice{i}\} \hspace*{10mm} \tr(y) \equiv \{\defchoice{y}\} \hspace*{10mm} \tr(j) \equiv \{\defchoice{j}\}\\\\
\hspace*{-55mm}\mbox{\t{Therefore, after substitution the result is:}}\\\\
\tr(x[i]<y[j]) \equiv \{\defchoice{\{\defchoice{\defchoice{x}\hspace*{0.5mm}[\hspace*{0.5mm}\defchoice{i}\hspace*{0.5mm}]}, x[\{i+1,i-1\}], y[i]\} < \{\defchoice{\defchoice{y}\hspace*{0.5mm}[\hspace*{0.5mm}\defchoice{j}\hspace*{0.5mm}]}, y[\{j+1,j-1\}], x[j]\}}\hspace*{0.5mm},\vspace*{2mm}\\ \hspace*{55mm}\{\{\defchoice{\defchoice{x}\hspace*{0.5mm}[\hspace*{0.5mm}\defchoice{i}\hspace*{0.5mm}]}, x[\{i+1,i-1\}], y[i]\}-1, 0\} < \{\{\defchoice{\defchoice{y}\hspace*{0.5mm}[\hspace*{0.5mm}\defchoice{j}\hspace*{0.5mm}]}, y[\{j+1,j-1\}], x[j]\}-1, 0\}\}\\
\end{array}
$
\caption{Application of $\tr_{\emod_1}$ (abbreviated $\tr$~) on expression $(x[i]<y[j])$.}
\figlabel{emapplication}
\end{figure*}

 The recursive steps of application of $\transExpr_{\emod_1}$ function
 on expression $(x[i] < y[j])$ are shown in
 \figref{emapplication}. This example illustrates two interesting
 features of the transformation function:
\begin{itemize}
{\item \bf Nested Transformations :} Once a rewrite rule $L \rightarrow R$
is applied to transform a program element matching $L$ to $R$, the
instructor may want to apply another rewrite rule on only a few
sub-terms of $R$. For example, she may want to avoid transforming the
sub-terms which have already been transformed by some other correction
rule. The $\eml$ language facilitates making such distinction between
the sub-terms for performing nested corrections using the $\prime$
(prime) operator. Only the sub-terms in $R$ that are tagged with the
prime operator are visited for applying further transformations (using
the $\transExpr_\emod$ method recursively on its tagged sub-terms
$t\dash$), whereas the remaining non-tagged sub-terms are not
transformed any further. After applying the rewrite rule $\crule_2$ in
the example, the sub-terms $x[i]$ and $y[j]$ are further transformed
by applying rewrite rules $\crule_1$ and $\crule_3$.

\item {\bf Ambiguous Transformations :} While transforming a program using
  an error model, it may happen that there are multiple rewrite rules
  that pattern match the program element $\wexpr$. After applying
  rewrite rule $\crule_2$ in the example, there are two rewrite rules
  $\crule_1$ and $\crule_3$ that pattern match the terms $x[i]$ and
  $y[j]$. After applying one of these rules ($\crule_1$ or
  $\crule_3$) to an expression $\var[\aexpr]$, we cannot apply the
  other rule to the transformed expression. In such ambiguous cases,
  the $\transExpr_\emod$ function creates a separate copy of the transformed
  program element ($\wexpr_i$) for each ambiguous choice and then
  performs the set union of all such elements to obtain the
  transformed program element. This semantics of handling ambiguity of
  rewrite rules also matches naturally with the intent of the
  instructor. If the instructor wanted to perform both transformations
  together on array accesses, she could have provided a combined
  rewrite rule such as $\var[\aexpr] \rightarrow ?\var[\{\aexpr+1, \aexpr-1\}]$.
\end{itemize}

\begin{theorem}
Given a well-formed error model $\emod$, the transformation function
$\transExpr_\emod$ always terminates.
\end{theorem}
\begin{proof}
From the definition of well-formed error model, each of its
constituent rewrite rule is also well-formed. Hence, each application
of a rewrite rule reduces the size of the syntax tree of terms that
are required to be visited further for transformation by
$\transExpr_\emod$. Therefore, the $\transExpr_\emod$ function
terminates in a finite number of steps.
\end{proof}

\section{Constraint-based Solving of $\implp$ programs}
\seclabel{algorithm} In the previous section, we saw the
transformation of an $\impl$ program to an $\implp$ program based on
an error model. We now present the translation of $\implp$ programs
into \Sk{} programs~\cite{sketchthesis}.

\subsection{Translation of $\implp$ programs to \Sk{}}

The $\implp$ programs are translated to \Sk{} programs to perform
constraint-based analysis. The main aspects of the translation
include the translation of : (i) python-like constructs in $\implp$ to \Sk{},
and (ii) set-expr choices in $\implp$ to \Sk{}
functions.

\paragraph{Handling dynamic typing of $\implp$ variables}
The dynamic typing of $\implp$ is handled using \t{MultiType} variable
as described in \secref{overview-translation}. The $\implp$
expressions and statements are transformed to \Sk{} functions that
perform the corresponding transformations over \t{MultiType}. For
example, the python statement \t{(a = b)} is translated to
\t{assignMT(a, b)}, where the \t{assignMT} function assigns
\t{MultiType} b to a. Similarly, the binary add expression \t{(a + b)}
is translated to \t{binOpMT(a, b, ADD\_OP)} that in turn calls the
function \t{addMT(a,b)} to add \t{a} and \t{b} as shown in
\figref{addmt}.

\begin{figure}[!htpb]
\begin{lstlisting}[language=C,mathescape=true]
MultiType addMT(MultiType a, MultiType b){
  assert a.flag == b.flag; // same types can be added
  if(a.flag == INTEGER)    // add for integers
   return new MultiType(val=a.val+b.val, flag = INTEGER);
  if(a.flag == LIST){      // add for lists
    int newLen = a.lst.len + b.lst.len;
    MultiType[newLen] newLVals = a.lst.lVals;
    for(int i=0; i<b.lst.len; i++)
      newLVals[i+a.lst.len] = b.lst.lVals[i];
  return new MultiType(lst = new MTList(lVals=newLVals, len=newLen), flag=LIST);}
  $\cdots$ $\cdots$
}
\end{lstlisting}
\caption{The \t{addMT} function for adding two \t{MultiType} \t{a} and \t{b}.}
\figlabel{addmt}
\end{figure}

\paragraph{Translation of $\implp$ set-expressions} The set-expressions in $\implp$ are translated to \Sk{} functions. The
function bodies obtained from translation ($\Phi$) of some of the
interesting $\implp$ constructs are shown in
\figref{fig-implp-trans}. The \Sk{} construct ??  (called \emph{hole})
is a placeholder for a constant value, which is filled up by the \Sk{}
synthesizer while solving the constraints to satisfy the given
specification.

The singleton sets consisting of an $\impl$ expression such as
$\{\aexpr\}$ are translated simply to the corresponding expression
itself. A set-expression of the form
$\{\defchoice{\saexpr_0},\cdots,\saexpr_n\}$ is translated recursively
to the $\sif$expression :$\sif(??)  ~\itos{\saexpr_0} ~\selse ~
\itos{\{\saexpr_1, \cdots, \saexpr_n\}}$, which means that the
synthesizer can optionally select the default set-expression
$\itos{\saexpr_0}$ (by choosing ?? to be \t{true}) or select one of
the other choices ($\saexpr_1,\cdots,\saexpr_n$). The set-expressions
of the form $\{\saexpr_0,\cdots,\saexpr_n\}$ are similarly translated
but with an additional statement for setting a fresh variable
$\choicek$ if the synthesizer selects the non-default choice
$\saexpr_0$.

The translation rules for the assignment statements ($\saexpr_0 :=
\saexpr_1$) results in $\sif$ expressions on both left and right sides
of the assignment. The $\sif$ expression choices occurring on the left
hand side are desugared to individual assignments. For example, the
left hand side expression $\sif(??) ~ x ~\selse ~ y := 10$ is
desugared to $\sif(??)~ x := 10 ~\selse ~ y := 10$. The infix
operators in $\implp$ are first translated to function calls and are
then translated to sketch using the translation for set-function
expressions. The remaining $\implp$ expressions are similarly
translated recursively as shown in the figure.
\nottechrep{
\begin{figure}
\begin{tabular}{c}
\begin{minipage}{.5\linewidth}
\begin{eqnarray*}
\itos{\{\aexpr\}} &=& \aexpr\\
\itos{\{\defchoice{\saexpr_0}, \cdots, \saexpr_n\}} & = & \sif(??)~ \itos{\saexpr_0} ~\selse ~\itos{\{\saexpr_1,\cdots,\saexpr_n\}} \\
\itos{\{\saexpr_0, \cdots, \saexpr_n\}} & = & \sif(??)~ \{\choicek=1; \itos{\saexpr_0}\}\\
& & \selse ~\itos{\{\saexpr_1,\cdots,\saexpr_n\}}\\
\itos{\saexpr_0[\saexpr_1]} & = & \itos{\saexpr_0}[\itos{\saexpr_1}]\\
\itos{\saexpr_0 = \saexpr_1} & = &  \itos{\saexpr_0} := \itos{\saexpr_1}\\
\end{eqnarray*}
\end{minipage}
\end{tabular}
\caption{The translation rules (shown partially) for converting $\implp$ set-exprs to corresponding \Sk{} function bodies.}
\figlabel{fig-implp-trans}
\end{figure}
}
\techrep{\begin{figure}
\begin{tabular}{c}
\begin{minipage}{.5\linewidth}
\begin{eqnarray*}
\itos{\{\aexpr\}} &=& \aexpr\\
\itos{\{\defchoice{\saexpr_0}, \cdots, \saexpr_n\}} & = & \sif(??)~ \itos{\saexpr_0} ~\selse ~\itos{\{\saexpr_1,\cdots,\saexpr_n\}} \\
\itos{\{\saexpr_0, \cdots, \saexpr_n\}} & = & \sif(??)~ \{\choicek=1; \itos{\saexpr_0}\}\\
& & \selse ~\itos{\{\saexpr_1,\cdots,\saexpr_n\}}\\
\itos{\saexpr_0[\saexpr_1]} & = & \itos{\saexpr_0}[\itos{\saexpr_1}]\\
\itos{\{f(\saexpr_1, \cdots, \saexpr_n)\}} & = &  f(\itos{\saexpr_1}, \cdots, \itos{\saexpr_n})\\
\Phi\{\defchoice{\fs_0},\cdots,\fs_n\}(\saexpr_1, & = &  \sif(??)~\itos{\fs_0(\saexpr_1, \cdots, \saexpr_n)}\\
 \saexpr_2,\cdots, \saexpr_n)& & \selse ~ \itos{\{\fs_1,\cdots,\fs_n\}(\saexpr_1, \cdots, \saexpr_n)}\\
\Phi\{\fs_0,\cdots,\fs_n\}(\saexpr_1, & = &  \sif(??)~\{ \choicek = 1; \\
 \saexpr_2 \cdots, \saexpr_n)& &  \hspace*{15mm}\itos{\fs_0(\saexpr_1, \cdots, \saexpr_n)}\}\\
& & \selse ~ \itos{\{\fs_1,\cdots,\fs_n\}(\saexpr_1, \cdots, \saexpr_n)}\\
\itos{\saexpr_0 = \saexpr_1} & = &  \itos{\saexpr_0} := \itos{\saexpr_1}\\
\itos{\ssexpr_0;\ssexpr_1} & = &  \itos{\ssexpr_0}; \itos{\ssexpr_1}\\
\itos{\t{if} ~ \sbexpr: ~ \ssexpr_0 ~ \t{else}: ~ \ssexpr_1} & = & \t{if} ~ (\itos{\sbexpr}) ~ \{\itos{\ssexpr_0}\} ~\t{else} ~ \{\itos{\ssexpr_1}\}\\
\itos{\t{while} ~ \sbexpr: ~ \ssexpr} & = & \t{while} ~ (\itos{\sbexpr}) ~ \{ \itos{\ssexpr} \}\\ 
\itos{\t{return} ~ \saexpr} & = & \t{return} ~~ \itos{\saexpr}
\end{eqnarray*}
\end{minipage}
\end{tabular}
\caption{The translation rules for converting $\implp$ set-exprs to corresponding \Sk{} function bodies.}
\figlabel{fig-implp-trans}
\end{figure}
}

\paragraph{Translating function calls}  The translation of function calls for recursive problems and for problems that require writing a
function that uses other sub-functions is parmeterized by three
options: 1) use the student's implementation of sub-functions, 2) use
the teacher's implementation of sub-functions, and 3) treat the
sub-functions as uninterpreted functions.

\paragraph{Generating the driver functions} The \Sk{} synthesizer supports checking equivalence of functions whose input arguments and return values are over \Sk{} primitive types such as \t{int}, \t{bit} and arrays. Therefore, after the translation of $\implp$ programs to \Sk{} programs, we need additional driver functions to 
integrate functions using \t{MultiType} input arguments and return
value to corresponding functions over \Sk{} primitive types. The
driver functions first converts the input arguments over primitive
types to corresponding \t{MultiType} variables using library functions
such as \t{computeMTFromInt}, and then calls the translated $\implp$
function with the \t{MultiType} variables. The returned \t{MultiType}
value is translated back to primitive types using library functions
such as \t{computeIntFromMT}. The driver function for student's
programs also consists of additional statements of the form
\t{if}($\choicek$) \t{totalCost++;} and the statement
\t{minimize(totalCost)}, which tells the synthesizer to compute a
solution to the Boolean variables $\choicek$ that minimizes the
\t{totalCost} variable.

\subsection{Incremental Solving for the Minimize hole expressions}
\seclabel{alg-increment-solving}
\begin{algorithm}
\caption{\t{CEGISMIN} Algorithm for Minimize expression}
\label{cegismin}
\begin{algorithmic}[1]
\State $\inpstate_0 \gets \inpstate_\t{random}, ~~~~i \gets 0, ~~~~\cnstrset_0 \gets \cnstrset, ~~~~\cnstrprev \gets \t{null}$
\While{\t{True}}
\State $i \gets i+1$
\State $\cnstrset_i \gets \t{Synth}(\inpstate_{i-1}, \cnstrset_{i-1})$ \Comment{Synthesis Phase}
\If{$\cnstrset_i = \t{UNSAT}$}
\Comment {Synthesis Fails}
\If{$\cnstrset_{\t{prev}} = \t{null}$} \Return \t{UNSAT\_SKETCH}
\Else ~~\Return PE(P,$\cnstrprev$)
\EndIf
\EndIf
\State $\textbf{choose} ~ \cnstr \in \cnstrset_i$
\State $\inpstate_i \gets \t{Verify}(\cnstr)$ \Comment{Verification Phase}
\If{$\inpstate_i = \t{null}$} \Comment {Verification Succeeds}
\State $(\t{minHole, minHoleValue}) \gets \t{getMinHoleValue}(\cnstr)$
\State $\cnstrprev \gets \cnstr$
\State $\cnstrset_i \gets \cnstrset_i \cup \{\t{encode}(\t{minHole} < \t{minHoleVal})\}$
\EndIf
\EndWhile
\hspace*{-19mm}
\end{algorithmic}
\end{algorithm}
We extend the \t{CEGIS} algorithm in \Sk{}~\cite{sketchthesis} to get
the \t{CEGISMIN} algorithm shown in Algorithm~\ref{cegismin} for
efficiently solving sketches that include a \t{minimize} hole
expression. The input state of the sketch program is denoted by
$\inpstate$ whereas the sketch constraint store denoted by
$\cnstrset$. Initially, the input state $\inpstate_0$ is assigned a
random input state value and the constraint store $\cnstrset_0$ is
assigned the constraint set obtained from the sketch program. The
variable $\cnstrprev$ stores the previous satisfiable hole values and
is initialized to \t{null}. In each iteration of the loop, the
synthesizer first performs the inductive synthesis phase where it
shrinks the constraints set $\cnstrset_{i-1}$ to $\cnstrset_{i}$ by
removing behaviors from $\cnstrset_{i-1}$ that do not conform to the
input state $\inpstate_{i-1}$. If the constraint set becomes
unsatisfiable, it either returns the sketch completed with hole values
from the previous solution if one exists, otherwise it returns
\t{UNSAT}. On the other hand, if the constraint set is satisfiable,
then it first chooses a conforming assignment to the hole values and
goes into the verification phase where it tries to verify the
completed sketch. If the verifier fails, it returns a
counter-example input state $\inpstate_i$ and the
synthesis-verification loop is repeated. If the verification phase
succeeds, instead of returning the result as is done in the \t{CEGIS}
algorithm, the \t{CEGISMIN} algorithm computes the value of
\t{minHole} from the constraint set $\cnstr$, stores the current
satisfiable hole solution $\cnstr$ in $\cnstrprev$, and adds an
additional constraint \{\t{minHole}<\t{minHoleVal}\} to the constraint
set $\cnstrset_i$. The synthesis-verification loop is then repeated
with this additional constraint to find a conforming value for the
\t{minHole} variable that is smaller than the current value in
$\cnstr$.

\subsection{Mapping \Sk{} solution to generate feedback}

Each correction rule in the error model $\fm$ is associated with a
feedback message, e.g. the integer variable initialization correction
rule $\var = \intN \rightarrow \var = \{\intN+1\}$ in the
\t{computeDeriv} error model is associated with the message
``Increment the right hand side of the initialization by $1$''. After
the \Sk{} synthesizer finds a solution to the constraints, the tool
maps back the values of unknown integer holes to their corresponding
expression choices. These expression choices are then mapped to
natural language feedback using the messages associated with the
corresponding correction rules, together with the line numbers.

\section{Implementation and Experiments}

We now briefly describe some of the implementation details of the tool,
and then describe the experiments we performed with it.

\subsection{Implementation and Features}
The tool's frontend that converts a python program to a \Sk{} program
is implemented in python itself and uses the python \t{ast} module for
parsing and rewriting ASTs. The backend system that solves the \Sk{}
program and provides feedback is implemented as a wrapper over the
\Sk{} system extended with the \t{CEGISMIN} algorithm. Error models in
our tool are currently written in python in terms of the python
AST. The tool also provides a mechanism to assign different cost
measure to correction rules in the error model to account for
different levels of mistakes.

\subsection{Benchmarks}

We created our benchmark set with problems taken from the Introduction
to Programming course at MIT (6.00) and the EdX version of the class
(6.00x). The \t{prodBySum} problem asks to compute the product using
the sum operator, the \t{oddTuples} problem asks to compute a list
consisting of alternate elements of the input list, the \t{evalPoly}
problem asks to compute the value of a polynomial on a given value,
the \t{iterPower} (and \t{recurPower}) problem asks to compute the
exponentiation using multiplication and the \t{iterGCD} problem
computes the gcd of two numbers. We also created a few AP-level
loop-over-arrays and dynamic programming
problems\footnote{\url{http://pexforfun.com/learnbeginningprogramming}}
on \pexfun to show the scalability and applicability of our technique
to other languages such as \csharp.

\subsection{Experiments}

\paragraph{Performance}
\tabref{results} shows the number of student attempts corrected for
each benchmark problem as well as the time taken by the tool to
provide the feedback. The experiments were performed on a 2.4GHz Intel
Xeon CPU with 16 cores and 16GB RAM. The experiments were performed
with bounds of 4 bits for input integer values and maximum length 4
for input lists. For each benchmark problem, we first removed the
student attempts with syntax errors to get the \t{Test Set} on which
we ran our tool. We then separated the attempts which were correct to
measure the effectiveness of the tool on the incorrect attempts. The
tool was able to provide appropriate corrections as feedback for
$65\%$ of all incorrect student attempts in around $10$ seconds on
average. The remaining 35\% of incorrect student attempts on which the
tool could not provide feedback fall in one of the following
categories:
\begin{itemize}
\item{\bf Completely incorrect solutions:} We observed many student
  attempts that were empty or performing trivial computations such as
  printing variables.
\item{\bf Big conceptual errors:} A common error we found in the case of
  \t{eval-poly-6.00x} was that a large fraction of incorrect attempts
  (260/541) were using the list function \t{index} to get the index of
  a value in the list, whereas the \t{index} function returns the
  index of first occurrence of the value in the list. Some other
  similar mistakes involve introducing or moving program statements
  from one place to another. These mistakes can not be corrected with
  the application of a set of local correction rules.
\item{\bf Unimplemented features:} Our implementation currently lacks a
  few of the complex python features such as pattern matching on list
  \t{enumerate} function.
\item{\bf Timeout:} In our experiments, we found less than $1\%$ of the
  student attempts timed out (set as 2 minutes).
\end{itemize}

\begin{figure*}[!htpb]
\begin{tabular}{ccc}
  \includegraphics[scale=0.25]{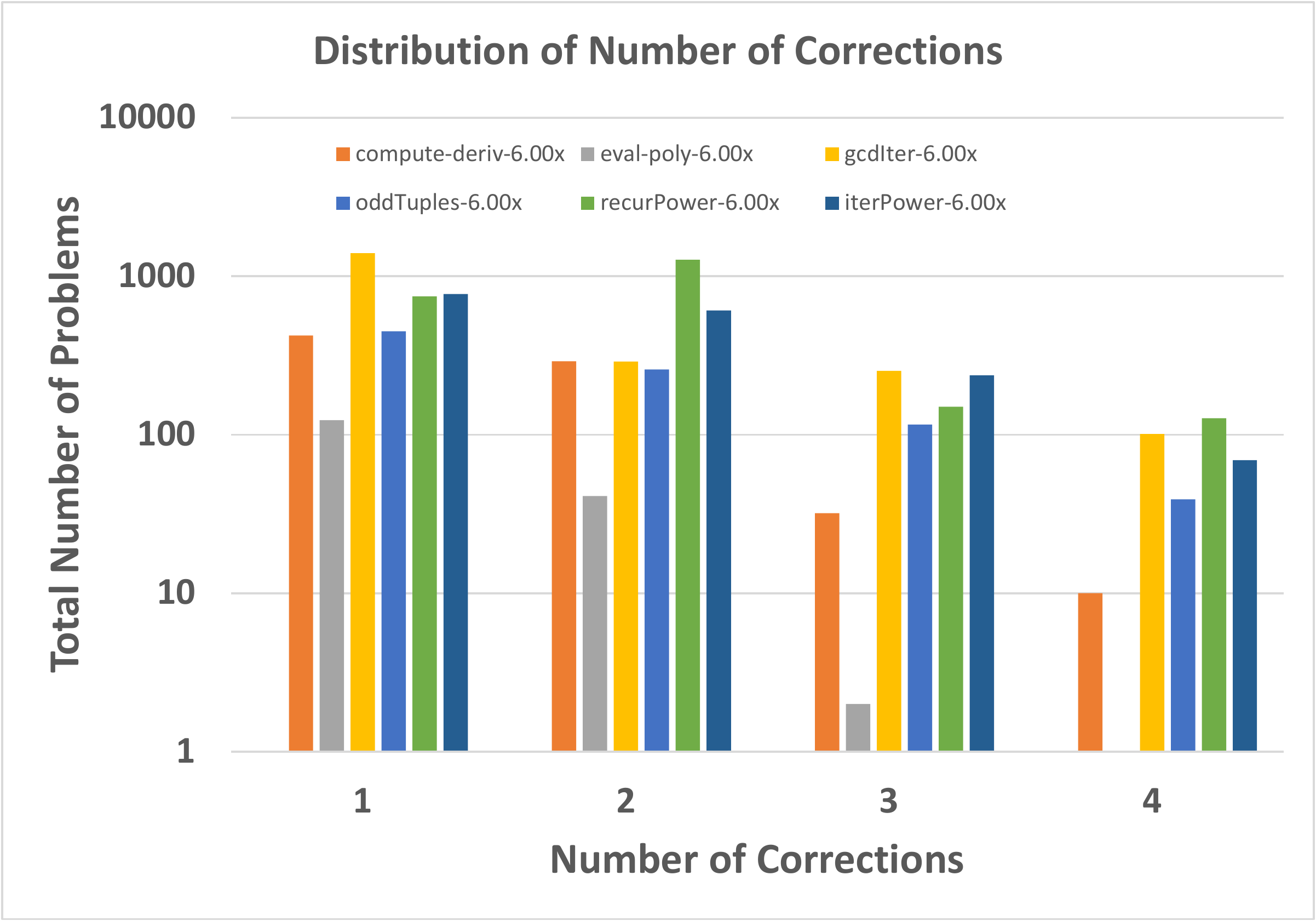}
&
\includegraphics[scale=0.24]{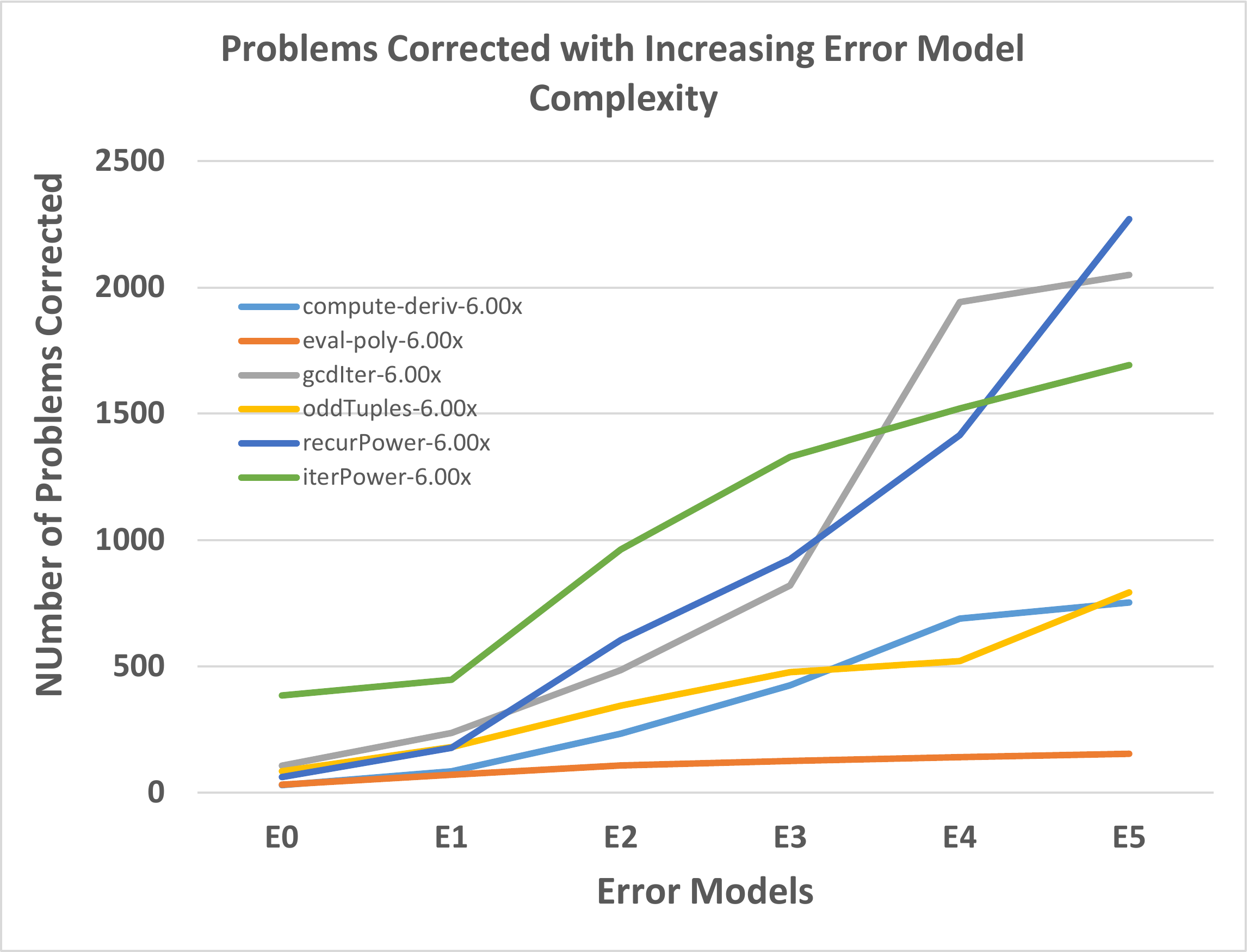}
&
\includegraphics[scale=0.28]{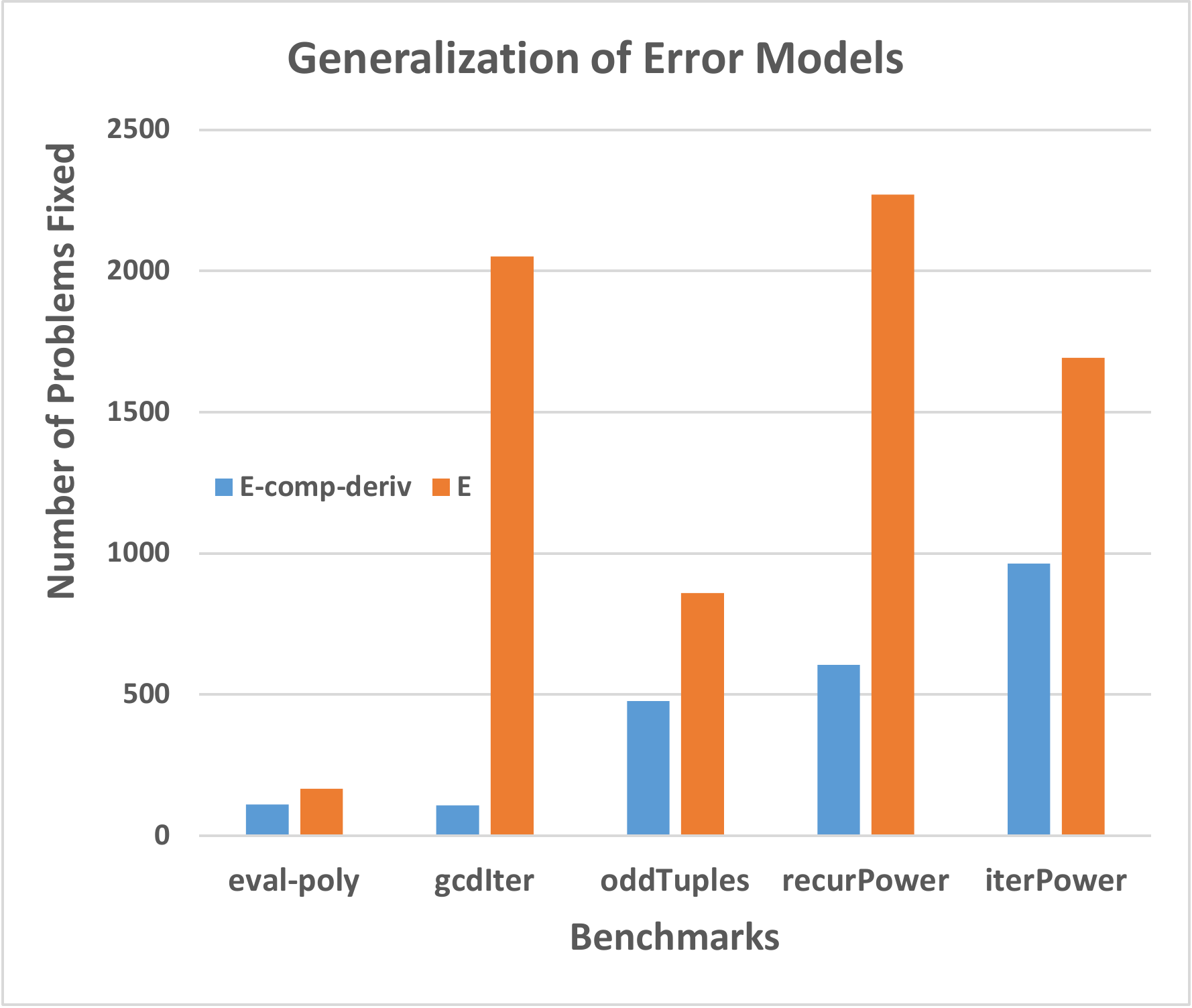}\\
(a) & (b) & (c)\\
\end{tabular}
\caption{(a) The number of incorrect problem that require different number of corrections, (b) the number of problems corrected with adding rules to the error models, and (c) the performance of compute-deriv error model on other problems.}
\figlabel{graphs}
\end{figure*}

\begin{table*}
\begin{center}
\begin{tabular}{|c|c|c|c|c|c|r|c|c|}
\hline
\multirow{2}{*}{\bf Benchmark} & {\bf Total} & {\bf Syntax} & \multirow{2}{*}{\bf Test Set} & \multirow{2}{*}{\bf Correct} &{\bf Incorrect} & {\bf Generated} & {\bf Average} & {\bf Median}\\
& {\bf Attempts} & {\bf Errors} & {} & {} &{\bf Attempts} & {\bf Feedback} & {\bf Time(in s)} & {\bf Time(in s)}\\
\hline
\t{prodBySum-6.00} & 1056 & 16 & 1040 & 772 & 268 & 218 (81.3\%)& 2.49s & 2.53s \\
\hline
\t{oddTuples-6.00} & 2386 & 1040 & 1346 & 1002 & 344 & 185 (53.8\%) &  2.65s & 2.54s\\
\hline
\t{computeDeriv-6.00} & 144& 20& 124& 21 & 103 & 88 (85.4\%)& 12.95s & 4.9s\\
\hline
\t{evalPoly-6.00} & 144 & 23 & 121 & 108 & 13 & 6 (46.1\%) & 3.35s & 3.01s\\
\hline
\t{computeDeriv-6.00x} & 4146 & 1134 & 3012 & 2094 & 918 & 753 (82.1\%)  & 12.42s & 6.32s\\
\hline
\t{evalPoly-6.00x} & 4698 & 1004 & 3694 & 3153 & 541 & 167 (30.9\%) & 4.78s & 4.19s\\
\hline
\t{oddTuples-6.00x} & 10985 & 5047 & 5938 &4182& 1756 & 860 (48.9\%)& 4.14s & 3.77s\\
\hline
\t{iterPower-6.00x} & 8982 & 3792 & 5190 & 2315& 2875 & 1693 (58.9\%) &3.58s & 3.46s\\
\hline
\t{recurPower-6.00x} & 8879 & 3395 & 5484 & 2546 & 2938 & 2271 (77.3\%) & 10.59s & 5.88s\\
\hline
\t{iterGCD-6.00x} & 6934 & 3732 & 3202 & 214  & 2988 & 2052 (68.7\%) & 17.13s & 9.52s\\
\hline
\hline
\t{stock-market-I(\csharp)} & 52 & 11 & 41 & 19 & 22 & 16 (72.3\%)& 7.54s & 5.23s\\
\hline
\t{stock-market-II(\csharp)} & 51 & 8 &  43 & 19 & 24 & 14 (58.3\%) & 11.16s & 10.28s\\
\hline
\t{restaurant rush (\csharp)} & 124 & 38 & 86 & 20 &  66 & 41 (62.1\%)& 8.78s & 8.19s\\
\hline
\end{tabular}
\end{center}
\caption{The percentage of student attempts corrected and the time taken for correction for the benchmark problems.}
\tablabel{results}
\end{table*}

\paragraph{Number of Corrections}
The number of student attempts that require different number of
corrections are shown in \figref{graphs}(a). We observe from the
figure that a large fraction of the problems require 3 and 4
coordinated corrections, and to provide feedback on such attempts, we
need a technology like ours that can symbolically encode the outcome
of different corrections on all input values.

\paragraph{Repetitive Mistakes}

In this experiment, we validate our hypothesis that students make
similar mistakes while solving a given problem. The graph in
\figref{graphs}(b) shows the number of student attempts corrected as
more rules are added to the error model of the benchmark problems. As
can be seen in the figure, adding a single rule to the error model can
lead to correction of hundreds of attempts.

\paragraph{Generalization of Error Models}

In this experiment, we check the hypothesis that whether the
correction rules generalize across problems of similar kind. The
result of running the \t{compute-deriv} error model on other benchmark
problems is shown in figref{graphs}(c). As expected, it does not
perform as well as the problem-specific error models, but it does
suffice to fix a fraction of the incorrect attempts and also provides
a good starting point to add more problem-specific rules.

\section{Related Work}
\seclabel{related}

In this section, we describe several related work to our technique
from the areas of automated programming tutors, automated program
repair, fault localization, automated debugging and program synthesis.

\subsection{AI based programming tutors}
There has been a lot of work done in the AI community for building
automated tutors for helping novice programmers learn programming by
providing feedback about semantic errors. These tutoring systems can
be categorized into the following two major classes:

{\bf Code-based matching approaches:} LAURA~\cite{laura} converts
teacher's and student's program into a graph based representation and
compares them heuristically by applying program transformations while
reporting mismatches as potential bugs. TALUS~\cite{talus} matches a
student's attempt with a collection of teacher's algorithms. It first
tries to recognize the algorithm used and then tentatively replaces
the top-level expressions in the student's attempt with the recognized
algorithm for generating correction feedback. The problem with these
approach is that the enumeration of all possible algorithms (with its
variants) for covering all corrections is very large and tedious on
part of the teacher.

{\bf Intention-based matching approaches:} LISP tutor~\cite{lisptutor}
creates a model of the student goals and updates it dynamically as the
student makes edits. The drawback of this approach is that it forces
students to write code in a certain pre-defined structure and limits
their freedom.  MENO-II~\cite{meno} parses student programs into a
deep syntax tree whose nodes are annotated with plan tags. This
annotated tree is then matched with the plans obtained from teacher's
solution. PROUST~\cite{proust}, on the other hand, uses a knowledge
base of goals and their corresponding plans for implementing them for
each programming problem. It first tries to find correspondence of
these plans in the student's code and then performs matching to find
discrepancies. CHIRON~\cite{chiron} is its improved version in which
the goals and plans in the knowledge base are organized in a
hierarchical manner based on their generality and uses machine
learning techniques for plan identification in the student code. These
approaches require teacher to provide all possible plans a student can
use to solve the goals of a given problem and do not perform well if
the student's attempt uses a plan not present in the knowledge base.

Our approach performs semantic equivalence of student's attempt and
teacher's solution based on exhaustive bounded symbolic verification
techniques and makes no assumptions on the algorithms or plans that
students can use for solving the problem. Moreover, our approach is
modular with respect to error models; the local correction rules are
provided in a declarative manner and their complex interactions are
handled by the solver itself.

\subsection{Automated Program Repair}

K$\ddot{o}$nighofer et. al.~\cite{prfmcad11} present an approach for
automated error localization and correction of imperative
programs. They use model-based diagnosis to localize components that
need to be replaced and then use a template-based approach for
providing corrections using SMT reasoning. Their fault model only
considers the right hand side (RHS) of assignment statements as
replaceable components. The approaches in~\cite{JobstmannGB05,staber}
frame the problem of program repair as a game between an environment
that provides the inputs and a system that provides correct values for
the buggy expressions such that the specification is satisfied. These
approaches only support simple corrections (e.g. correcting RHS side
of expressions) in the fault model as they aim to repair large
programs with arbitrary errors. In our setting, we exploit the fact
that we have access to the dataset of previous student mistakes that
we can use to construct a \emph{concise and precise} error model. This
enables us to model more sophisticated transformations such as
introducing new program statements, replacing LHS of assignments
etc. in our error model. Our approach also supports minimal cost
changes to student's programs where each error in the model is
associated with a certain cost, unlike the earlier mentioned
approaches.

Mutation-based program repair~\cite{mutationicst} performs mutations
repeatedly to statements in a buggy program in order of their
suspiciousness until the program becomes correct. The large state
space of mutants ($10^{12}$) makes this approach infeasible. Our
approach uses a symbolic search for exploring correct solutions over
this large set. There are also some genetic programming
approaches that exploit redundancy present in other parts of the
code for fixing faults~\cite{arcuri,weimer}. These techniques are not
applicable in our setting as such redundancy is not present in
introductory programming problems.

\subsection{Automated Debugging and Fault localization}

Techniques like Delta Debugging ~\cite{deltadebugging} and
QuickXplain~\cite{quickexplain} aim to simplify a failing test case
to a minimal test case that still exhibits the same failure. Our
approach can be complemented with these techniques to restrict the
application of rewrite rules to certain failing parts of the program
only. There are many algorithms for fault localization~\cite{fl1,fl2}
that use the difference between faulty and successful executions of
the system to identify potential faulty locations. Jose
et. al.~\cite{rupakpldi11} recently suggested an approach that uses a
MAX-SAT solver to satisfy maximum number of clauses in a formula
obtained from a failing test case to compute potential error
locations. These approaches, however, only localize faults for a
single failing test case and the suggested error location might not be
the desired error location, since we are looking for common error
locations that cause failure of multiple test cases. Moreover, these
techniques provide only a limited set of suggestions (if any) for
repairing these faults.

\subsection{Program Synthesis}

Program synthesis has been used recently for many applications such as
synthesis of efficient low-level code~\cite{solar2005bitstreaming,
  completefunctionalsynthesis}, inference of efficient synchronization
in concurrent programs~\cite{ibmAGS}, snippets of excel
macros~\cite{cacm12}, relational data
structures~\cite{hawkins11,hawkins12} and angelic
programming~\cite{angelicpopl}. The
\Sk~tool~\cite{solar2005bitstreaming,sketchthesis} takes a partial
program and a reference implementation as input and uses
constraint-based reasoning to synthesize a complete program that is
equivalent to the reference implementation. In general cases, the
template of the desired program as well as the reference specification
is unknown and puts an additional burden on the users to provide them;
in our case we use the student's solution as the template program and
teacher's solution as the reference implementation. A recent work by
Gulwani et. al.~\cite{sumitgeometry} also uses program synthesis
techniques for automatically synthesizing solutions to ruler/compass
based geometry construction problems. Their focus is primarily on
finding a solution to a given geometry problem whereas we aim to
provide feedback on a given programming exercise solution. \techrep{A
comprehensive survey on various program synthesis techniques and
different intention mechanisms for specification can be found
in~\cite{sumitsurvey}.}

\section{Conclusions}
\seclabel{conclusion} In this paper, we presented a new technique of
automatically providing feedback for introductory programming
assignments that can complement manual and test-cases based
techniques. The technique uses an error model describing the potential
corrections and constraint-based synthesis to compute minimal
corrections to student's incorrect solutions. We have evaluated our
technique on a large set of benchmarks and it can correct $65\%$ of
incorrect solutions. We believe this technique can provide a basis for
providing automated feedback to hundreds of thousands of students
learning from online introductory programming courses that are being
taught by MITx and Udacity.

\bibliographystyle{abbrvnat}
\bibliography{references}

\begin{thebibliography}{30}
\providecommand{\natexlab}[1]{#1}
\providecommand{\url}[1]{\texttt{#1}}
\expandafter\ifx\csname urlstyle\endcsname\relax
  \providecommand{\doi}[1]{doi: #1}\else
  \providecommand{\doi}{doi: \begingroup \urlstyle{rm}\Url}\fi

\bibitem[tec()]{techrep}
Automated feedback generation for introductory programming assignments.
\newblock Technical report.
\newblock {supplementary material}.

\bibitem[Adam and Laurent(1980)]{laura}
A.~Adam and J.-P.~H. Laurent.
\newblock {LAURA, A System to Debug Student Programs.}
\newblock \emph{Artif. Intell.}, 15\penalty0 (1-2):\penalty0 75--122, 1980.

\bibitem[Arcuri(2008)]{arcuri}
A.~Arcuri.
\newblock On the automation of fixing software bugs.
\newblock In \emph{ICSE Companion}, 2008.

\bibitem[Ball et~al.(2003)Ball, Naik, and Rajamani]{fl1}
T.~Ball, M.~Naik, and S.~K. Rajamani.
\newblock From symptom to cause: localizing errors in counterexample traces.
\newblock In \emph{POPL}, 2003.

\bibitem[Bod\'{\i}k et~al.(2010)Bod\'{\i}k, Chandra, Galenson, Kimelman, Tung,
  Barman, and Rodarmor]{angelicpopl}
R.~Bod\'{\i}k, S.~Chandra, J.~Galenson, D.~Kimelman, N.~Tung, S.~Barman, and
  C.~Rodarmor.
\newblock Programming with angelic nondeterminism.
\newblock In \emph{POPL}, 2010.

\bibitem[Debroy and Wong(2010)]{mutationicst}
V.~Debroy and W.~Wong.
\newblock Using mutation to automatically suggest fixes for faulty programs.
\newblock In \emph{ICST}, 2010.

\bibitem[Ertmer et~al.(2007)Ertmer, Richardson, Belland, Camin, Connolly,
  Coulthard, Lei, and Mong]{peerfeedback}
P.~Ertmer, J.~Richardson, B.~Belland, D.~Camin, P.~Connolly, G.~Coulthard,
  K.~Lei, and C.~Mong.
\newblock Using peer feedback to enhance the quality of student online
  postings: An exploratory study.
\newblock \emph{Journal of Computer-Mediated Communication}, 12\penalty0
  (2):\penalty0 412--433, 2007.

\bibitem[Farrell et~al.(1984)Farrell, Anderson, and Reiser]{lisptutor}
R.~G. Farrell, J.~R. Anderson, and B.~J. Reiser.
\newblock An interactive computer-based tutor for lisp.
\newblock In \emph{AAAI}, 1984.

\bibitem[Forrest et~al.(2009)Forrest, Nguyen, Weimer, and Goues]{weimer}
S.~Forrest, T.~Nguyen, W.~Weimer, and C.~L. Goues.
\newblock A genetic programming approach to automated software repair.
\newblock In \emph{GECCO}, 2009.

\bibitem[Groce et~al.(2006)Groce, Chaki, Kroening, and Strichman]{fl2}
A.~Groce, S.~Chaki, D.~Kroening, and O.~Strichman.
\newblock Error explanation with distance metrics.
\newblock \emph{STTT}, 8\penalty0 (3):\penalty0 229--247, 2006.

\bibitem[Gulwani et~al.(2008)Gulwani, Srivastava, and Venkatesan]{GulwCBA}
S.~Gulwani, S.~Srivastava, and R.~Venkatesan.
\newblock Program analysis as constraint solving.
\newblock In \emph{PLDI}, 2008.

\bibitem[Gulwani et~al.(2011)Gulwani, Korthikanti, and Tiwari]{sumitgeometry}
S.~Gulwani, V.~A. Korthikanti, and A.~Tiwari.
\newblock Synthesizing geometry constructions.
\newblock In \emph{PLDI}, 2011.

\bibitem[Gulwani et~al.(2012)Gulwani, Harris, and Singh]{cacm12}
S.~Gulwani, W.~R. Harris, and R.~Singh.
\newblock Spreadsheet data manipulation using examples.
\newblock In \emph{{CACM}}, 2012.

\bibitem[Hawkins et~al.(2011)Hawkins, Aiken, Fisher, Rinard, and
  Sagiv]{hawkins11}
P.~Hawkins, A.~Aiken, K.~Fisher, M.~C. Rinard, and M.~Sagiv.
\newblock Data representation synthesis.
\newblock In \emph{PLDI}, 2011.

\bibitem[Hawkins et~al.(2012)Hawkins, Aiken, Fisher, Rinard, and
  Sagiv]{hawkins12}
P.~Hawkins, A.~Aiken, K.~Fisher, M.~C. Rinard, and M.~Sagiv.
\newblock Concurrent data representation synthesis.
\newblock In \emph{PLDI}, 2012.

\bibitem[Jobstmann et~al.(2005)Jobstmann, Griesmayer, and Bloem]{JobstmannGB05}
B.~Jobstmann, A.~Griesmayer, and R.~Bloem.
\newblock Program repair as a game.
\newblock In \emph{CAV}, pages 226--238, 2005.

\bibitem[Johnson and Soloway(1985)]{proust}
W.~L. Johnson and E.~Soloway.
\newblock Proust: Knowledge-based program understanding.
\newblock \emph{IEEE Trans. Software Eng.}, 11\penalty0 (3):\penalty0 267--275,
  1985.

\bibitem[Jose and Majumdar(2011)]{rupakpldi11}
M.~Jose and R.~Majumdar.
\newblock Cause clue clauses: error localization using maximum satisfiability.
\newblock In \emph{PLDI}, 2011.

\bibitem[Junker(2004)]{quickexplain}
U.~Junker.
\newblock {QUICKXPLAIN: preferred explanations and relaxations for
  over-constrained problems}.
\newblock In \emph{AAAI}, 2004.

\bibitem[K\"{o}nighofer and Bloem(2011)]{prfmcad11}
R.~K\"{o}nighofer and R.~P. Bloem.
\newblock Automated error localization and correction for imperative programs.
\newblock In \emph{FMCAD}, 2011.

\bibitem[Kuncak et~al.(2010)Kuncak, Mayer, Piskac, and
  Suter]{completefunctionalsynthesis}
V.~Kuncak, M.~Mayer, R.~Piskac, and P.~Suter.
\newblock Complete functional synthesis.
\newblock PLDI, 2010.

\bibitem[Murray(1987)]{talus}
W.~R. Murray.
\newblock Automatic program debugging for intelligent tutoring systems.
\newblock \emph{Computational Intelligence}, 3:\penalty0 1--16, 1987.

\bibitem[Sack et~al.(1992)Sack, Soloway, and Weingrad]{chiron}
W.~Sack, E.~Soloway, and P.~Weingrad.
\newblock From {PROUST} to {CHIRON}: Its design as iterative engineering:
  Intermediate results are important!
\newblock In \emph{In J.H. Larkin and R.W. Chabay (Eds.), Computer-Assisted
  Instruction and Intelligent Tutoring Systems: Shared Goals and Complementary
  Approaches.}, pages 239--274, 1992.

\bibitem[Solar-Lezama(2008)]{sketchthesis}
A.~Solar-Lezama.
\newblock \emph{Program Synthesis By Sketching}.
\newblock PhD thesis, EECS Dept., UC Berkeley, 2008.

\bibitem[Solar-Lezama et~al.(2005)Solar-Lezama, Rabbah, Bodik, and
  Ebcioglu]{solar2005bitstreaming}
A.~Solar-Lezama, R.~Rabbah, R.~Bodik, and K.~Ebcioglu.
\newblock Programming by sketching for bit-streaming programs.
\newblock In \emph{PLDI}, 2005.

\bibitem[Soloway et~al.(1981)Soloway, Woolf, Rubin, and Barth]{meno}
E.~Soloway, B.~P. Woolf, E.~Rubin, and P.~Barth.
\newblock {Meno-II: An Intelligent Tutoring System for Novice Programmers}.
\newblock In \emph{IJCAI}, 1981.

\bibitem[Srivastava et~al.(2010)Srivastava, Gulwani, and Foster]{SrivPts}
S.~Srivastava, S.~Gulwani, and J.~Foster.
\newblock From program verification to program synthesis.
\newblock \emph{POPL}, 2010.

\bibitem[Staber et~al.(2005)Staber, Jobstmann, and Bloem]{staber}
S.~S. Staber, B.~Jobstmann, and R.~P. Bloem.
\newblock Finding and fixing faults.
\newblock In \emph{Correct Hardware Design and Verification Methods}, Lecture
  notes in computer science, pages 35 -- 49, 2005.

\bibitem[Vechev et~al.(2010)Vechev, Yahav, and Yorsh]{ibmAGS}
M.~Vechev, E.~Yahav, and G.~Yorsh.
\newblock Abstraction-guided synthesis of synchronization.
\newblock In \emph{POPL}, New York, NY, USA, 2010. ACM.

\bibitem[Zeller and Hildebrandt(2002)]{deltadebugging}
A.~Zeller and R.~Hildebrandt.
\newblock Simplifying and isolating failure-inducing input.
\newblock \emph{IEEE Transactions on Software Engineering}, 28:\penalty0
  183--200, 2002.

\end{thebibliography}


\begin{thebibliography}{26}
\providecommand{\natexlab}[1]{#1}
\providecommand{\url}[1]{\texttt{#1}}
\expandafter\ifx\csname urlstyle\endcsname\relax
  \providecommand{\doi}[1]{doi: #1}\else
  \providecommand{\doi}{doi: \begingroup \urlstyle{rm}\Url}\fi

\bibitem[Adam and Laurent(1980)]{laura}
A.~Adam and J.-P.~H. Laurent.
\newblock {LAURA, A System to Debug Student Programs.}
\newblock \emph{Artif. Intell.}, 15\penalty0 (1-2):\penalty0 75--122, 1980.

\bibitem[Arcuri(2008)]{arcuri}
A.~Arcuri.
\newblock On the automation of fixing software bugs.
\newblock In \emph{ICSE Companion}, pages 1003--1006, 2008.

\bibitem[Ball et~al.(2003)Ball, Naik, and Rajamani]{fl1}
T.~Ball, M.~Naik, and S.~K. Rajamani.
\newblock From symptom to cause: localizing errors in counterexample traces.
\newblock In \emph{POPL}, pages 97--105, 2003.

\bibitem[Bod\'{\i}k et~al.(2010)Bod\'{\i}k, Chandra, Galenson, Kimelman, Tung,
  Barman, and Rodarmor]{angelicpopl}
R.~Bod\'{\i}k, S.~Chandra, J.~Galenson, D.~Kimelman, N.~Tung, S.~Barman, and
  C.~Rodarmor.
\newblock Programming with angelic nondeterminism.
\newblock In \emph{POPL}, pages 339--352, 2010.

\bibitem[Debroy and Wong(2010)]{mutationicst}
V.~Debroy and W.~Wong.
\newblock Using mutation to automatically suggest fixes for faulty programs.
\newblock In \emph{Software Testing, Verification and Validation (ICST), 2010
  Third International Conference on}, pages 65 --74, april 2010.

\bibitem[Farrell et~al.(1984)Farrell, Anderson, and Reiser]{lisptutor}
R.~G. Farrell, J.~R. Anderson, and B.~J. Reiser.
\newblock An interactive computer-based tutor for lisp.
\newblock In \emph{AAAI}, pages 106--109, 1984.

\bibitem[Forrest et~al.(2009)Forrest, Nguyen, Weimer, and Goues]{weimer}
S.~Forrest, T.~Nguyen, W.~Weimer, and C.~L. Goues.
\newblock A genetic programming approach to automated software repair.
\newblock In \emph{GECCO}, pages 947--954, 2009.

\bibitem[Groce et~al.(2006)Groce, Chaki, Kroening, and Strichman]{fl2}
A.~Groce, S.~Chaki, D.~Kroening, and O.~Strichman.
\newblock Error explanation with distance metrics.
\newblock \emph{STTT}, 8\penalty0 (3):\penalty0 229--247, 2006.

\bibitem[Gulwani(2010)]{sumitsurvey}
S.~Gulwani.
\newblock Dimensions in program synthesis.
\newblock In \emph{PPDP}, pages 13--24, 2010.

\bibitem[Gulwani(2011)]{excelpopl}
S.~Gulwani.
\newblock Automating string processing in spreadsheets using input-output
  examples.
\newblock In \emph{POPL}, pages 317--330, 2011.

\bibitem[Gulwani et~al.(2011{\natexlab{a}})Gulwani, Jha, Tiwari, and
  Venkatesan]{sumitbitvector}
S.~Gulwani, S.~Jha, A.~Tiwari, and R.~Venkatesan.
\newblock Synthesis of loop-free programs.
\newblock In \emph{PLDI}, pages 62--73, 2011{\natexlab{a}}.

\bibitem[Gulwani et~al.(2011{\natexlab{b}})Gulwani, Korthikanti, and
  Tiwari]{sumitgeometry}
S.~Gulwani, V.~A. Korthikanti, and A.~Tiwari.
\newblock Synthesizing geometry constructions.
\newblock In \emph{PLDI}, pages 50--61, 2011{\natexlab{b}}.

\bibitem[Itzhaky et~al.(2010)Itzhaky, Gulwani, Immerman, and Sagiv]{sumitgraph}
S.~Itzhaky, S.~Gulwani, N.~Immerman, and M.~Sagiv.
\newblock A simple inductive synthesis methodology and its applications.
\newblock In \emph{OOPSLA}, pages 36--46, 2010.

\bibitem[Jobstmann et~al.(2005)Jobstmann, Griesmayer, and Bloem]{JobstmannGB05}
B.~Jobstmann, A.~Griesmayer, and R.~Bloem.
\newblock Program repair as a game.
\newblock In \emph{CAV}, pages 226--238, 2005.

\bibitem[Johnson and Soloway(1985)]{proust}
W.~L. Johnson and E.~Soloway.
\newblock Proust: Knowledge-based program understanding.
\newblock \emph{IEEE Trans. Software Eng.}, 11\penalty0 (3):\penalty0 267--275,
  1985.

\bibitem[Jose and Majumdar(2011)]{rupakpldi11}
M.~Jose and R.~Majumdar.
\newblock Cause clue clauses: error localization using maximum satisfiability.
\newblock In \emph{Proceedings of the 32nd ACM SIGPLAN conference on
  Programming language design and implementation}, PLDI '11, pages 437--446,
  New York, NY, USA, 2011. ACM.

\bibitem[Junker(2004)]{quickexplain}
U.~Junker.
\newblock {QUICKXPLAIN: preferred explanations and relaxations for
  over-constrained problems}.
\newblock In \emph{Proceedings of the 19th national conference on Artifical
  intelligence}, AAAI'04, pages 167--172, 2004.

\bibitem[K\"{o}nighofer and Bloem(2011)]{prfmcad11}
R.~K\"{o}nighofer and R.~P. Bloem.
\newblock Automated error localization and correction for imperative programs.
\newblock In \emph{Proceedings of 11th International Conference 2011 Formal
  Methods in Computer Aided Design (FMCAD 2011)}, 2011.

\bibitem[Kuncak et~al.(2010)Kuncak, Mayer, Piskac, and
  Suter]{completefunctionalsynthesis}
V.~Kuncak, M.~Mayer, R.~Piskac, and P.~Suter.
\newblock Complete functional synthesis.
\newblock In \emph{Proceedings of the 2010 ACM SIGPLAN conference on
  Programming language design and implementation}, PLDI '10, pages 316--329,
  2010.

\bibitem[Murray(1987)]{talus}
W.~R. Murray.
\newblock Automatic program debugging for intelligent tutoring systems.
\newblock \emph{Computational Intelligence}, 3:\penalty0 1--16, 1987.

\bibitem[Sack et~al.(1992)Sack, Soloway, and Weingrad]{chiron}
W.~Sack, E.~Soloway, and P.~Weingrad.
\newblock From {PROUST} to {CHIRON}: Its design as iterative engineering:
  Intermediate results are important!
\newblock In \emph{In J.H. Larkin and R.W. Chabay (Eds.), Computer-Assisted
  Instruction and Intelligent Tutoring Systems: Shared Goals and Complementary
  Approaches.}, pages 239--274, 1992.

\bibitem[Solar-Lezama(2008)]{sketchthesis}
A.~Solar-Lezama.
\newblock \emph{Program Synthesis By Sketching}.
\newblock PhD thesis, EECS Dept., UC Berkeley, 2008.

\bibitem[Solar-Lezama et~al.(2005)Solar-Lezama, Rabbah, Bodik, and
  Ebcioglu]{solar2005bitstreaming}
A.~Solar-Lezama, R.~Rabbah, R.~Bodik, and K.~Ebcioglu.
\newblock Programming by sketching for bit-streaming programs.
\newblock In \emph{PLDI '05: Proceedings of the 2005 ACM SIGPLAN conference on
  Programming language design and implementation}, pages 281--294, 2005.

\bibitem[Soloway et~al.(1981)Soloway, Woolf, Rubin, and Barth]{meno}
E.~Soloway, B.~P. Woolf, E.~Rubin, and P.~Barth.
\newblock {Meno-II: An Intelligent Tutoring System for Novice Programmers}.
\newblock In \emph{IJCAI}, pages 975--977, 1981.

\bibitem[Staber et~al.(2005)Staber, Jobstmann, and Bloem]{staber}
S.~S. Staber, B.~Jobstmann, and R.~P. Bloem.
\newblock Finding and fixing faults.
\newblock In \emph{Correct Hardware Design and Verification Methods}, Lecture
  notes in computer science, pages 35 -- 49, 2005.

\bibitem[Zeller and Hildebrandt(2002)]{deltadebugging}
A.~Zeller and R.~Hildebrandt.
\newblock Simplifying and isolating failure-inducing input.
\newblock \emph{IEEE Transactions on Software Engineering}, 28:\penalty0
  183--200, 2002.

\end{thebibliography}

\end{document}